\documentclass[english,twocolumn,transaction]{IEEEtran}
\usepackage[T1]{fontenc}
\usepackage{float}
\usepackage{amsmath}
\usepackage{amsthm}
\usepackage{amssymb}
\usepackage{stackrel}
\usepackage{graphicx}
\usepackage{setspace}
\usepackage{url}
\usepackage{algorithm}
\usepackage{algorithmic}
\usepackage{booktabs}
\usepackage{stfloats}
\usepackage{cite}
\usepackage{float}
\usepackage{bm}

\usepackage{subfigure}
\makeatletter

\floatstyle{ruled}
\newfloat{algorithm}{tbp}{loa}
\providecommand{\algorithmname}{Algorithm}
\floatname{algorithm}{\protect\algorithmname}

\theoremstyle{plain}

\theoremstyle{definition}

\theoremstyle{plain}

\theoremstyle{plain}

\newcommand{\RNum}[1]{\uppercase\expandafter{\romannumeral #1\relax}}

\IEEEoverridecommandlockouts
\usepackage{ifpdf}
\usepackage{cite}
\ifCLASSOPTIONcompsoc
\usepackage[caption=false,font=normalsize,labelfont=sf,textfont=sf]{subfig}
\else
\usepackage[caption=false,font=footnotesize]{subfig}
\fi
\usepackage{amsmath}
\@ifundefined{showcaptionsetup}{}{%
 \PassOptionsToPackage{caption=false}{subfig}}
\usepackage{subfig}
\makeatother

\usepackage{colortbl}

\usepackage{babel}
\newtheorem{define}{Definition}
\newtheorem{property}{Property}

\renewcommand\figurename{Fig.}
\usepackage{amsmath}
\allowdisplaybreaks[4]

\begin{document}

\title{
Placement and Resource Allocation of Wireless-Powered Multiantenna UAV for Energy-Efficient Multiuser NOMA}

\author{Zhongyu Wang,~\emph{Student~Member,~IEEE},~Tiejun~Lv,~\emph{Senior Member,~IEEE},\\~Jie~Zeng,~\emph{Senior Member,~IEEE},~and~Wei~Ni,~\emph{Senior~Member,~IEEE}
\thanks{This work was supported by the National Natural Science Foundation of China (No. 62001264), and the Natural Science Foundation of Beijing (No. L192025).
\emph{(Corresponding author: Tiejun Lv and Jie Zeng).}}

\thanks{
Z. Wang and T. Lv are with the School of Information and Communication Engineering, Beijing University of Posts and Telecommunications (BUPT), Beijing 100876, China (e-mail: \{zhongyuwang, lvtiejun\}@bupt.edu.cn).

J. Zeng is with the School of Cyberspace Science and Technology, Beijing Institute of Technology, Beijing 100081, China (e-mail: zengjie.2002@tsinghua.org.cn).

W. Ni is with Data61, Commonwealth Scientific and Industrial Research,
Sydney, NSW 2122, Australia (e-mail: wei.ni@data61.csiro.au).
}}
\maketitle

\begin{abstract}

This paper investigates a new downlink nonorthogonal multiple access (NOMA) system, where a multiantenna unmanned aerial vehicle (UAV) is powered by wireless power transfer (WPT) and serves as the base station for multiple pairs of ground users (GUs) running NOMA in each pair.
An energy efficiency (EE) maximization problem is formulated to jointly optimize the WPT time and the placement for the UAV, and the allocation of the UAV's transmit power between different NOMA user pairs and within each pair.
To efficiently solve this nonconvex problem, we decompose the problem into three subproblems using block coordinate descent.
For the subproblem of intra-pair power allocation within each NOMA user pair, we construct a supermodular game with confirmed convergence to a Nash equilibrium.
Given the intra-pair power allocation, successive convex approximation is applied to convexify and solve the subproblem of WPT time allocation and inter-pair power allocation between the user pairs.
Finally, we solve the subproblem of UAV placement by using the Lagrange multiplier method. Simulations show that our approach can substantially outperform its alternatives that do not use NOMA and WPT techniques or that do not optimize the UAV location.
\end{abstract}


\section{Introduction }
\renewcommand\figurename{Fig.}
\subsection{Background and Motivation}
Sixth-generation (6G) communication networks are anticipated to
enhance network data availability, the mobile data
rate, and ubiquitous connectivity for proliferative Internet-of-Things (IoT) terminals \cite{sym12040676}.
In occasional events (e.g., disasters, concerts, and sports events),
conventional networks with rigid structures are not suitable \cite{8811579}.
Unmanned aerial vehicles (UAVs), also known as drones, have attracted extensive attention for their fast deployment and excellent flexibility \cite{9453824,8411465}.
Integration of UAVs into future wireless networks has become a popular
research area, with emphases on two aspects \cite{7463007}.
The first aspect pertains to the services that UAVs can provide to wireless networks.
UAVs can serve as mobile aerial base stations (BSs) or relays to support ground users (GUs) and provide   line-of-sight (LoS) for air-to-ground communications.
UAVs can provide fast network deployment, increase network capacity, and deliver connectivity to blindspots, e.g., for  disaster rescue \cite{7486987}.
The second aspect pertains to the services that wireless networks can provide to UAVs.
UAVs can be aerial users and served through cellular technologies, such as long-term evolution (LTE) and 5G \cite{azam2020energy}. 

A critical challenge faced by UAV-assisted wireless communications is the limited on-board energy of UAVs, affecting the endurance and longevity of UAV operations \cite{8663615}.   
By optimizing the elevation of a hovering UAV, a considerable amount of UAV energy is saved
due to the presence of communication--friendly LoS between the UAV and GU \cite{8885517}. 
Wireless power transfer (WPT) is another viable solution to prolong the flight time
of UAVs \cite{8995773}. 
To support power transmission over a long distance, resonant
beams were used to charge UAVs \cite{7589757}. 
In \cite{8995773}, the authors proposed a rotary-wing UAV-enabled mobile relaying system in which a laser was used to charge the UAV relay. 
In \cite{9209946}, a multi-UAV location-aware wireless-powered communication network was considered, where WPT-powered UAVs were used to monitor the environment and report to ground micro-BSs. 

Spectrum efficiency (SE) and energy efficiency (EE) are important measures for wireless communication systems \cite{8766136}. 
Nonorthogonal multiple access (NOMA) is a promising technique to improve the SE and EE \cite{7842433}. 
Compared with conventional orthogonal multiple access (OMA) which each user with a dedicated time or frequency resource, multiple users' signals are superposed in the power domain under NOMA \cite{8269066}. 
Successive interference cancellation (SIC) was adopted to decode superposed signals \cite{7273963}. 
A cooperative NOMA scheme was proposed to mitigate the interference and maximize the weighted sum rate of the UAV and GUs in an uplink UAV-assisted NOMA system by jointly optimizing the UAV transmit power and rate \cite{8641388}.   
Nasir \emph{et al.} \cite{8672190} employed NOMA in a UAV-aided multiuser communication system and solved the max-min rate optimization problem subject to the total power, bandwidth, UAV altitude, and antenna beamwidth. 
In \cite{8449221}, a multiantenna UAV generated directional beams and simultaneously served multiple users, where the sum rate was maximized by employing NOMA and beam scanning.

Existing studies of NOMA-based UAV communication systems have typically assumed that single-antenna UAVs hover at fixed positions and communicate with GUs without considering UAV placement \cite{8488592,9162133,8641388,9094731}. 
The trajectory of a single-antenna UAV was optimized to maximize the EE in \cite{7888557}. 
The transmit power and trajectory of the UAV were jointly optimized to maximize the minimum average throughput in \cite{8489918}.   
Wang \emph{et al.} \cite{8255570} treated a single-antenna UAV as a mobile access point to serve multiple GUs and exploited the mobility of the UAV to maximize the system throughput. 
Multiantenna techniques have been increasingly applied to UAV-assisted communication systems \cite{10.3389/frcmn.2021.696111}.
Multiantenna beamforming can improve the data throughput by exploiting spatial diversity and multiplexing, and thus potentially help reduce the propulsion energy consumption of the UAVs \cite{9055113}.
To the best of our knowledge, no existing studies have considered multiantenna UAVs supporting NOMA.  
In addition, most existing studies on UAV-enabled WPT systems under NOMA transmission have adopted directional radio frequency (RF) signals to (re)charge GUs under the assumption that the UAVs have abundant energy.

Distinctively different from the existing studies, we consider that a WPT-powered multiantenna UAV acts as an aerial BS and serves multiple pairs of users running NOMA with each pair.  We maximize the EE of the UAV's transmissions by optimizing the placement of the UAV, the allocation of the transmit power between different user pairs and within each user pair, and the WPT time between a power beacon and the UAV.
The key contributions of this paper are summarized as follows.

\begin{itemize}

 \item
 We construct a UAV-assisted NOMA network, where a WPT-powered, multiantenna UAV serves as the BS for multiple pairs of GUs running NOMA.
 We derive the total throughput of all user pairs and the energy consumption model for the power beacon, UAV, and GUs.
 Subsequently, the EE model is defined for the system. 

 \item
An EE maximization problem is formulated, which jointly optimizes the placement of the UAV, the allocation of the transmit power between different user pairs and within each user pair, and the WPT time.
The problem is nonconvex with coupled variables. We solve the problem efficiently by decomposing it into three subproblems:
1) intra-pair power allocation within each user pair;
2) WPT time allocation, and inter-pair power allocation between user pairs;
and 3) UAV location design. 

 \item
 A noncooperative game is constructed to solve the subproblem of intra-pair power allocation within each user pair, and a supermodular game-based algorithm is proposed with confirmed convergence to a Nash equilibrium.
 Given the intra-pair power allocation, successive convex approximation (SCA) is
invoked to solve the subproblem of the WPT time and inter-pair power allocation between user pairs. 
Then, we optimize the UAV's hovering location using the Lagrange multiplier method.

\end{itemize}


The remainder of this paper is organized as follows.
In Section II, the WPT-powered, UAV-assisted, multiuser NOMA system model is described and the considered problem is formulated.
In Section III, we decompose the problem into three subproblems and elaborate on the solutions to the subproblems.
Section IV shows simulation results.
Conclusions are drawn in Section V.

\section{System Model and Problem Formulation}

\subsection{System Model}

As shown in Fig. 1, we consider a UAV-assisted NOMA communication system in which a UAV equipped with $M$ antennas, harvests energy from a power beacon, and serves $2N$ single-antenna GUs.
The UAV is located at $(x_{0},y_{0},h)$, where $h$ is a constant altitude for safety considerations.
The GUs are distributed in a disc-shaped area with a radius of $R$, which is expected to be covered by the UAV.
The $2N$ GUs are divided into $N$ pairs.
The $N$ users closer to the coverage center are referred to as ``cell-center users'' (CC users), while the remaining $N$ users are ``cell-edge users'' (CE users).
The UAV employs NOMA to pair a CC user with a CE user \footnote{
Generally speaking, running SIC at the receiver of each user incurs additional complexity, depending on the number of users multiplexed in the same channel.
The hardware complexity and processing delay increase with the number of users in each channel. Moreover, the SIC can suffer from increasingly severe error propagations with a growing number of users per channel.
As a result, a channel is often assigned to no more than two users in practice \cite{7273963}.
}.
At a given altitude, the mean path loss between the user and the UAV is an increasing function of their horizontal distance \cite{8301585}.
The users can be ordered according to their distances from the coverage center, satisfying $R>{{d}_{N,2}}>{{d}_{N-1,2}}>\cdots >{{d}_{1,2}}>{{d}_{N,1}}>{{d}_{N-1,1}}>\cdots >{{d}_{1,1}}$.
The mean path losses of the users take the same order. Like Strategy 2 developed in \cite{8641425}, a two-user pairing scheme is adopted, where the closest CC user to the coverage center is paired with the further CE user.
The second-closest CC user is paired with the second-farthest CE user, and so on and so forth, until all users are paired.
This pairing scheme is illustrated in Fig. 2.
For each pair, the CC user can cancel the interference of the CE user by using SIC.

\begin{figure}[t]
\center
\includegraphics[width=9cm]{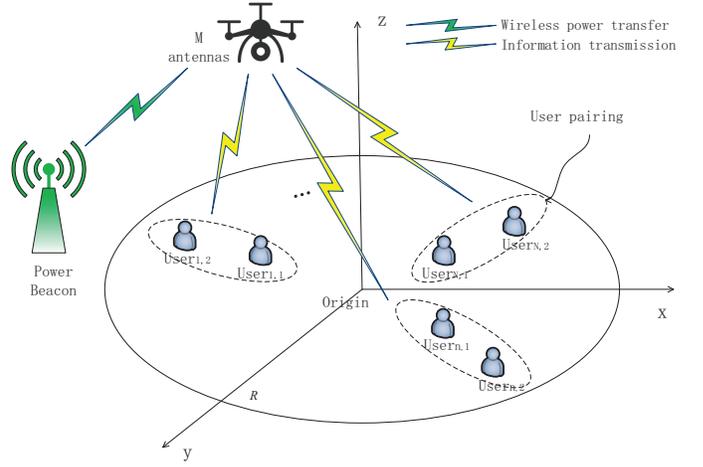}
\caption{UAV-based NOMA communication network with WPT. }
\end{figure}

\begin{figure}[t]
\center
\includegraphics[width=9cm]{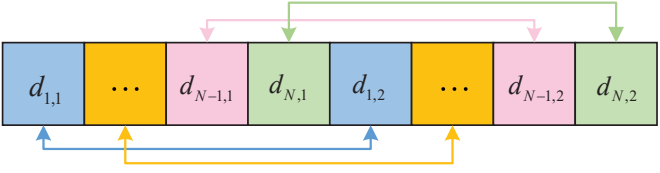}
\caption{Two-user NOMA pairing scheme for an even number of users. }
\end{figure}

We assume that a wireless-powered UAV is continuously available during the service time within its coverage and uses the harvest-then-transmit protocol to sustain its services \cite{9204681}. Given the self-sustainability of the wireless-powered UAV, it can remain at the locations as long as required or until the traffic distributions served by the UAVs change substantially.
There could be transitional states during which the traffic distribution changes and the UAV needs fly to new positions.
Consider the fast flight speed of the UAV and hence relatively negligible transitional states.
The UAV divides a communication cycle of $ T$ into two stages.
First, the power beacon transmits the RF energy to the UAV for time $\tau T$.
By using the harvested energy, the UAV transmits signals to all $N$ pairs of CC and CE users using downlink NOMA for time $(1-\tau )T$.
Without loss of generality, $T$ is normalized to 1.

In the WPT processing stage, the channel between the power beacon and the UAV is usually modelled as a LOS link, as given by
\begin{equation}\label{eq0}
  \mathbf{H}_{PU}= \sqrt{M^2\beta_{PU}} \mathbf{a}_{T}\left(\theta_{PU}\right) \mathbf{a}_{R}^{H}(\theta_{PU})
\end{equation}
where $\beta_{PU} = \beta_{0}d_{PU}^{-\beta}$ is the path loss.
$d_{PU} = \sqrt{(x_b-x_0)^2+(y_b-y_0)^2+h^2}$ is the distance between the power beacon and the UAV with the coordinates of the power beacon $(x_b, y_b, 0)$ being outside the coverage.
$\beta_{0}$ is the path loss at the reference distance of 1 m.
$\beta$ is the path loss exponent.
$\mathbf{a}_{T}\left(\theta_{PU}\right) =\frac{1}{\sqrt{M}}\left[1, e^{j 2 \pi \frac{d}{\lambda} \cos \theta_{PU}}, \ldots, e^{j 2 \pi \frac{(M-1) d}{\lambda} \cos \theta_{PU}}\right]^{T}$ is the array response of the transmitter and $\mathbf{a}_{R}\left(\theta_{PU}\right) =\frac{1}{\sqrt{M}}\left[1, e^{j 2 \pi \frac{d}{\lambda} \cos \theta_{PU}}, \ldots, e^{j 2 \pi \frac{(M-1) d}{\lambda} \cos \theta_{PU}}\right]^{T}$ is the array response of the receiver.
Here, $\lambda$ is the signal wavelength, $d$ is the spacing between adjacent elements, and $\theta_{PU}$ is the azimuth angle-of-departure (AOD) of the power beacon, and $\theta_{PU} = \text{arcsin} (h/d_{PU})$.

After the energy harvesting process, the transmit power of the UAV is given by
\begin{equation}\label{eq1}
  P_T=\tau \|\mathbf{H}_{PU}\mathbf{W}_{PU}\|^2  P \xi/(1-\tau),
\end{equation}
where $\mathbf{H}_{PU}$ is the channel between the power beacon and the UAV;
$\mathbf{W}_{PU}$ is the energy beamforming vector used to maximize the harvested energy; $\xi$ is the energy conversion efficiency; and $P$ is the transmit power of the power beacon.

The location of a user (i.e., a CC or CE user) is $(x_{n,k},y_{n,k},z_{n,k})$.
User $k$ ($k = 1 $ for a CC user or $k = 2 $ for a CE user) of the $n$-th pair is denoted by
User$_{n,k}$.
The distance between the UAV and User$_{n,k}$ is given by
\begin{equation}\label{eq3}
  d_{n,k} = \sqrt{(x_0-x_{n,k})^2+(y_0-y_{n,k})^2+(h-z_{n,k})^2}.
\end{equation}

Given the elevation $h$, we expect that the channels between the UAV and users are dominated by the LoS.
The channel gain between the UAV and User$_{n,k}$ is given by
\begin{equation}\label{eq2}
  \mathbf{h}_{n,k}=\sqrt{M \beta_{n,k}} \mathbf{a}\left(\theta_{n,k}\right),
\end{equation}
where $\beta_{n,k}$ denotes large-scale fading, $\mathbf{a}\left(\theta_{n,k}\right)$ is the array response, $\theta_{n,k}$ is the AOD of the user, and $\theta_{n,k} = \text{arcsin} (h/d_{n,k})$.
The large-scale path loss is
$\beta_{n,k} = \beta_{0}d_{n,k}^{-\beta}$ with $\beta_{0}$ being the channel power gain between the UAV and User$_{n,k}$.
Considering a uniform linear array (ULA) and a LoS channel, $\mathbf{a}\left(\theta_{n,k}\right)=\frac{1}{\sqrt{M}}\left[1, e^{j 2 \pi \frac{d}{\lambda} \cos \theta_{n,k}}, \ldots, e^{j 2 \pi \frac{(M-1) d}{\lambda} \cos \theta_{n,k}}\right]^{T}$, where $\lambda$ is the signal wavelength and $d$ is the spacing between adjacent antenna elements.

At the UAV, the transmit signal vector for the user pairs, after superposition coding, can be expressed as
\begin{equation}
\begin{aligned}
\mathbf{s} &=\left[\sqrt{\alpha_{1,1}} s_{1,1}+\sqrt{\alpha_{1,2}} s_{1,2}, \ldots, \sqrt{\alpha_{N, 1}} s_{ N,1}+\sqrt{\alpha_{ N,2}} s_{ N,2}\right]^{T} \\
&=\left[s_{1}, \ldots, s_{N}\right]^{T},
\end{aligned}
\end{equation}
where $s_{n,1}$ and $s_{n,2}$ denote the signals defined for the CC user and CE user in the $n$-th pair ($n = 1,\cdots,N$), respectively; and $\alpha_{n,1}$ and $\alpha_{n,2}$ are the corresponding power allocation coefficients, satisfying $\alpha_{n,1}+\alpha_{n,2} = 1$.
$s_n = \sqrt{\alpha_{n,1}} s_{n,1}+\sqrt{\alpha_{n,2}} s_{n,2}$. The transmit signals are precoded to reduce inter-pair interference, as given by
 \begin{equation}\label{eq5}
   \mathbf{x}=\mathbf{Q} \mathbf{s}=\left[\mathbf{q}_{1}, \ldots, \mathbf{q}_{n}, \ldots, \mathbf{q}_{N}\right] \mathbf{s},
 \end{equation}
where $\mathbf{Q}\in\mathbb{C}^{M\times N}$ is the precoding matrix and $\mathbf{q}_{n}\in\mathbb{C}^{M\times 1}$ is the precoding vector for $s_n$.
The received signal at the $k$-th ($k = 1, 2$) user of the $n$-th pair can be written as
\begin{equation}\label{eq6}
  {y}_{n, k}=\mathbf{h}_{n, k}^{H} \mathbf{q}_{n} s_{n}+\sum_{j=1, j \neq n}^{N} \mathbf{h}_{n, k}^{H} \mathbf{q}_{j} s_{j}+{o}_{n, k},
\end{equation}
where $\mathbf{h}_{n, k}\in\mathbb{C}^{M\times 1}$ is the channel state vector of the user;
$\sum_{j=1, j \neq n}^{N} \mathbf{h}_{n, k}^{H} \mathbf{q}_{j} s_{j}$ is the received interference signal of all pairs except the $n$-th pair;
and ${o}_{n, k}$ is the additive white Gaussian noise (AWGN) at the $k$-th user of the $n$-th pair, and  obeys the zero-mean unit-variance complex symmetric Gaussian distribution, i.e., ${o}_{n, k}\sim \mathbb{CN}(0,1)$.

The user in each user pair, e.g., User$_{n,k}$, detects the received signal ${y}_{n, k}$ using a combining variable ${w}_{n, k}$. Then, the postprocessed signal is given by
\begin{equation}\label{eq7}
{w}_{n, k}^{H} {y}_{n, k}={w}_{n, k}^{H} \mathbf{h}_{n, k}^{H}  \mathbf{Q} \mathbf{s}+ {w}_{n, k}^{H} {o}_{n, k}.
\end{equation}
Following the NOMA principle, the power allocation coefficients yield $\alpha_{n,1} \leq \alpha_{n,2}$ $(n\in {1, \cdots, N})$ so that CC User$_{n,1}$ can successfully decode CE User$_{n,2}$'s signal first and then subtract it from the superposition coded signal via SIC.

Inter-pair interference would decrease the received signal-to-interference-plus-noise ratio (SINR) and increase the system outage probability; therefore, we propose an interference cancellation method by constructing an interference cancellation combining matrix at the UAV to cancel the inter-pair interference for all CC users.
At the receiver end,
we take the widely used equal gain combining ${w}_{n,1}$ as the combining scheme of the CC user in each pair because of its low complexity.
As a result, the effective channel vector of the CC user is $\mathbf{g}_{n,1}={w}_{n,1}\mathbf{h}_{n,1}\in \mathbb{C}^{M\times 1}$.
In the system, precoding can help separate user data streams in different directions. Linear precoding, such as zero-forcing (ZF) precoding, MMSE precoding, and heuristic MMSE precoding \cite{christopoulos2015user,8641436,9048676,6876818}, is known for low complexity and easy implementation.
With a balanced consideration of precoding gain and complexity at a resource-restricted WPT-powered UAV, we choose ZF precoding because it is one of the most used linear precoding schemes and known to achieve full spatial multiplexing and multiuser diversity gains with a low complexity \cite{6425521}.
The precoding matrix can suppress the inter-pair interference among the CC users in all user pairs by satisfying
\begin{equation}\label{eq8}
  \left[\mathbf{g}_{n,1}^{H}, \ldots, \mathbf{g}_{n-1, 1}^{H}, \mathbf{g}_{n+1, 1}^{H}, \ldots, \mathbf{g}_{N, 1}^{H}\right]^{H} \mathbf{q}_{n}=\mathbf{0},
\end{equation}
where $\mathbf{q}_{n}$ can be obtained by singular value decomposition (SVD).
We can obtain the precoding matrix as $\mathbf{Q} =\left[\mathbf{q}_{1}, \ldots, \mathbf{q}_{n}, \ldots, \mathbf{q}_{N}\right]$.
The matrix $\mathbf{Q}$ eliminates the inter-pair interference since $\mathbf{g}_{n,1} \mathbf{q}_{m}={w}_{n,1}\mathbf{h}_{n,1}\mathbf{q}_{m} (n\neq m)$.
With the precoding at the UAV and the postprocessing of (\ref{eq6}) and (\ref{eq7}) at the CC users, the receive signal-to-noise ratio (SNR) of the CC user in the $n$-th pair for decoding its own signal is given by
\begin{equation}\label{eq9}
   \gamma_{n,1}^{CC} = {\alpha_{n,1} p_n |\mathbf{h}_{n,1}^{H}\mathbf{q}_{n}|^2},
\end{equation}
where $p_n$ is the transmit power of the UAV allocated to the $n$-th pair of the GUs.\\
\indent
We also construct the receive combining vectors for CE users.
Different from the inter-pair interference elimination for the CC user, the inter-pair interference of the CE user in each pair is suppressed with the receive combining vector.
Inspired by maximal ratio combining \cite{4100918}, the combining vector of the CE user in the $n$-th pair, i.e., $w_{n,2}$, is given by
\begin{equation}\label{eq10}
  w_{n,2} = \mathbf{h}_{n,2}^{H}\mathbf{v}_{n,2},
\end{equation}
where $\mathbf{v}_{n,2}\in\mathbb{C}^ { M \times 1}$. According to (\ref{eq7}) and (\ref{eq10}), the received signal of the CE user in the $n$-th pair is given by
\begin{equation}\label{eq11}
\begin{aligned}
&w_{n,2}^H y_{n,2} = \mathbf{v}_{n,2}^H \mathbf{h}_{n,2}\mathbf{h}_{n,2}^{H} \mathbf{q}_{n} s_{n}\\
&+ \sum_{m=1, m \neq n}^{N} \mathbf{v}_{n,2}^{H} \mathbf{h}_{n,2}\mathbf{h}_{n,2}^{H} \mathbf{q}_{m} s_{m}+ w_{n,2}^H n_{n,2}.
\end{aligned}
\end{equation}
The effective channel vector of the CE user in the $n$-th pair is $\mathbf{g}_{n,2} = \mathbf{h}_{n,2} \mathbf{h}_{n,2}^{H} \mathbf{q}_{n} \in \mathbb{C}^{M\times 1}$. Similarly, to mitigate the inter-pair interference of the CE users and increase the SINR of the CE users, the proposed combining variable $w_{2,n}$ needs to satisfy the following conditions:
\begin{equation}\label{eq12}
  \mathbf{v}_{n,2}^H\left[\mathbf{g}_{n,2}, \ldots, \mathbf{g}_{n-1, 2}, \mathbf{g}_{n+1, 2}, \ldots, \mathbf{g}_{N, 2}\right]=0.
\end{equation}
The vector $\mathbf{v}_{n,2}$ is the left singular vector of the effective matrix and can be obtained from the null space of the effective matrix. According to (\ref{eq12}), we construct the receive combining variable $w_{n,2}$ for each CE user.
The receive SINR of the CE user, ${\gamma}_{n,2}^{CE}$, is given by
\begin{equation}\label{eq13}
  {\gamma}_{n,2}^{CE}=\frac{\alpha_{n,2} p_{n}\left| \mathbf{h}_{n,2}^{H} \mathbf{q}_{n}\right|^{2}}{ 1 +\alpha_{n,1} p_{n}\left|  \mathbf{h}_{n,2}^{H} \mathbf{q}_{n}\right|^{2}}.
\end{equation}\\
\indent
According to (\ref{eq9}) and (\ref{eq13}), the data rates of the CC and CE users of the $n$-th pair are
\begin{equation}\label{eq-add1}
  R_{n,1}^{CC} = B\log _{2}\left(1+\gamma_{n,1}^{CC}\right)
\end{equation}
\begin{equation}\label{eq-add2}
  R_{n,2}^{CE} = B\log _{2}\left(1+\gamma_{n,2}^{CE}\right).
\end{equation}

The sum throughput of the considered system is given by
\begin{small}
\begin{equation}\label{eq14}
\begin{aligned}
& R_{sum} =  \sum_{i=1}^{N} \mathbb{E}\left\{(T-\tau)(R_{n,1}^{CC}+R_{n,2}^{CE})\right\}  \\
&= (T-\tau)B\sum_{i=1}^{N} \mathbb{E}\left\{\log _{2}\left(1+ {\alpha_{n,1} p_n |\mathbf{h}_{n,1}^{H} \mathbf{q}_{n}|^2 }  \right)\right.  \\
&\left.+\log _{2}\left(1+\frac{\alpha_{n,2} p_{n}\left| \mathbf{h}_{n,2}^{H} \mathbf{q}_{n}\right|^{2}}{1+\alpha_{n,1} p_{n}\left| \mathbf{h}_{n,2}^{H} \mathbf{q}_{n}\right|^{2}}\right)\right\}.
\end{aligned}
\end{equation}
\end{small}
\\
\indent The energy consumption of the system consists of three parts: The power beacon, the UAV energy consumption, and the energy consumption of all CC and CE users, as given by
\begin{equation}\label{eq15}
  E_{sum} = P_{pb} \tau + P_{uav}T + 2NP_{user}T,
\end{equation}
where $P_{pb}$ is the power consumption of the power beacon, consisting of all the RF power consumption $MP_a$ of the power beacon, with $P_{pb}=MP_a$.
$P_a$ is the power consumption of each RF chain.
$P_{uav}$ is the power consumption of the UAV, consisting of two parts.
The first part accounts for transmissions, i.e., the transmit power of signal processing $P_T$ and all the RF chain power consumption of the UAV $MP_a$.
The second part is the propulsion energy to keep the UAV aloft and support its mobility.
The propulsion power consumption model of rotary-wing UAVs flying at a speed of $V$ is
$P_{prop}(V)=P_{0}\left(1+\frac{3 V^{2}}{U_{tip}^{2}}\right)+P_{i}\left(\sqrt{1+\frac{V^{4}}{4 v_{0}^{4}}}-\frac{V^{2}}{2 v_{0}^{2}}\right)^{\frac{1}{2}}
+\frac{d_{0} \rho s A V^{3}}{2}$,
where $P_0$ and $P_i$ are two constants related to the physical properties of the UAV and the flight environment, such as the weight, rotor radius and air density;
$U_{tip}$ is the tip speed of the rotor blade;
$v_0$ is the mean rotor-induced velocity during hovering;
$d_0$ and $s$ are the fuselage drag ratio and rotor solidity, respectively;
and $\rho$ and $A$ are the air density and rotor disc area, respectively.
The power required for hovering is $P_{prop}(0)$.
Thus, $P_{uav}=P_T+MP_a+P_{prop}(0) $.
$P_{user}$ is the energy consumption of each user.

Based on (\ref{eq14}) and (\ref{eq15}), the EE of the considered system is written as
\begin{equation}\label{eq16}
  \Xi_{EE} = \frac{R_{sum}}{E_{sum}}.
\end{equation}
\subsection{Problem Formulation}
The WPT time, transmit power, power allocation coefficient, and placement of the multiantenna UAV are jointly optimized to maximize the EE, i.e.,
\begin{subequations}
\label{equP1}
\begin{align}
& \ \ \ \ \mathbf{P 1}:\ \ \underset{\{{p_n}\}_{n=1}^{N}, \alpha_{n,k}, {\tau},{x_0},{y_0} }{\mathop{\max }} \ \ \Xi_{EE} \\
& s.t.\mathbf{C1}:\ \ \ 0\le \sum\nolimits_{n=1}^{N}{{p}_{n}}\le {{P}_{T}}, \ \  \forall n \in \mathcal{N} \\
&\ \ \ \ \mathbf{C2}:\ \ \ 0\le \tau \le T  \\
&\ \ \ \ \mathbf{C3}:\ \ \ R_{n,1}^{CC} \geq R_{\min}, \forall n \in \mathcal{N} \\
&\ \ \ \ \mathbf{C4}:\ \ \ R_{n,2}^{CE} \geq R_{\min}, \forall n \in \mathcal{N} \\
&\ \ \ \ \mathbf{C5}:\ \ \ x_{\min } \leq x_{0} \leq x_{\max } \\
&\ \ \ \ \mathbf{C6}:\ \ \ y_{\min } \leq y_{0} \leq y_{\max }\\
&\ \ \ \ \mathbf{C7}:\ \ \ \Sigma_{k=1}^{2}\alpha_{n,k}=1, \forall n \in \mathcal{N}\\
&\ \ \ \ \mathbf{C8}:\ \ \ \sum_{n=1}^{N} p_{n} \mathbf{q}_{n}^{H} \mathbf{q}_{n} \leq P_{\iota}, \iota=1, \ldots N,
\end{align}
\end{subequations}
\noindent
where constraint $\mathbf{C1}$ specifies the maximum transmit power of the UAV;
$\mathbf{C2}$ gives the time range of harvesting energy;
$\mathbf{C3}$ specifies the minimum rate requirement of the CC user, with $R_{\min}$ being the minimum rate;
$\mathbf{C4}$ provides the minimum rate requirement of the CE user;
$\mathbf{C5}$ and $\mathbf{C6}$ specify the ranges of the $x$- and $y$-coordinates of the UAV location, respectively, with $x_{\min}$ being the minimum of $x_0$, $x_{\max}$ the maximum of $x_0$,
$y_{\min}$ the minimum of $y_0$, $y_{\max}$ the maximum of $y_0$;
$\mathbf{C7}$ is self-explanatory;
and $\mathbf{C8}$ specifies the per-antenna transmit power constraint (PAC) with $P_{\iota}$ being the transmit power limit of each antenna.

\section{Resource Allocation and UAV Placement}

The optimal solution to problem (\ref{equP1}) requires a joint optimization of the transmit power, power allocation coefficients, time allocation, and UAV location design, which is a nonconvex problem.
We decouple the problem into three subproblems.
We first optimize the power allocation coefficients within each user pair by developing a game-theoretic approach.
Then, the power for each user pair and WPT time are allocated using the SCA method.
Finally, we apply the Lagrange multiplier method to solve the UAV placement subproblem.
The three subproblems are iteratively solved using block coordinate descent (BCD).

\subsection{Game-Theoretic Power Allocation Coefficient Optimization per Pair}
Many existing power allocation methods use fractional transmit power allocation \cite{8790780,7982784} and exhaustive search \cite{AnxinLI2015}.
Exhaustive search discretizes and enumerates possible power allocation coefficients for a user pair, resulting in a prohibitively high complexity.
The fractional transmit power allocation method depends on the channel-to-noise-ratios (CNRs) \cite{8790780} of the CC and CE users, and cannot readily apply to NOMA systems with multiantenna transmitters.
In this section, we invoke a game-theoretic approach \cite{7723834} to obtain the power allocation coefficients of the $n$-th user pair.

For notational simplicity, we define ${h}_{n,1} = \mathbf{h}_{n,1}^{H}\mathbf{q}_{n}$ and ${h}_{n,2} = \mathbf{h}_{n,2}^{H}\mathbf{q}_{n}$.
Given fixed UAV location and inter-pair power allocation, and WPT time allocation (e.g., equal inter-pair power allocation and equal time allocation for WPT and wireless information transfer at initialization), the optimization of the power allocation coefficients for the CC and CE users of all the user pairs can be decomposed into $N$ subproblems as follows. \vspace{-1.5em}
\begin{subequations}
\label{equP2}
\begin{align}
&\ \ \ \ \ \ \mathbf{P 2}:\underset{ \alpha_{n,k}} {\mathop{\max }}\ \ \ \ \Pi \\
& s.t.\ \ \mathbf{C1}:\Sigma_{k=1}^{2}\alpha_{n,k}=1, \ \  \forall n \in \mathcal{N},
\end{align}
\end{subequations}
where $\Pi=\frac{ B\log _{2}\left(1+p_{n,1} |{h}_{n,1}|^2 \right) }{ p_{n,1}+ P_{uav} + 2NP_{user} }\text{+}\frac{B \log _{2}\left(1+\frac{p_{n,2} |{h}_{n,2}|^2}{1 + p_{n,1} |{h}_{n,2}|^2 }\right)} { p_{n,2}+ P_{uav} + 2NP_{user} }$.

In the considered system, multiple users can be multiplexed at the same time and frequency.
Consider that each user is selfish and aims to assign the corresponding power to maximize its own EE.
We model the $N$ subproblems as noncooperative games.
Let ${\mathcal{U}}_n=\{1,2\}$ denote a two-user set consisting of the CC and CE users in each pair
\footnote{It is possible to construct the Nash equilibrium involving more than
three users (if needed). By definition, a supermodular game can be characterized by strategic
complementarities \cite{7723834}. In other words, when one player takes a more aggressive action, then
the rest of the users would follow the same behavior. The authors of \cite{John1996} investigated the theory of $N$-player non-cooperative game, and proved that any limited non-cooperative game has at least one equilibrium point. }.
We denote the power allocation strategy for each user in the $n$-th pair as the set $\mathcal{P}_n=\{\mathcal{P}_{n,1}, \mathcal{P}_{n,2}\}$, where $\mathcal{P}_{n,1}=\{p_{n,1}\mid p_{n,1}+p_{n,2}<p_{n}, p_{n,1}\geq 0\}$ and $\mathcal{P}_{n,2}=\{p_{n,2}\mid p_{n,1}+p_{n,2}<p_{n}, p_{n,2}\geq 0\}$.
We decouple the original problem into $N$ independent games $\mathcal{G}_n = [{\mathcal{U}}_n, {\mathcal{P}}_n, {\mathcal{F}}_n]$, where
${\mathcal{F}}_n = \{f_{n,1}, f_{n,2}\}$ denotes the set of utility functions.
Assuming that the channel gain of the CC user is higher than that of the CE user, i.e., $|{h}_{n,1}|>|{h}_{n,2}|$, the utility functions of the CC and CE users are defined as
\begin{align}
f_{n,1} &= \frac{ \log _{2}\left(1+p_{n,1} |{h}_{n,1}|^2 \right) }{ p_{n,1}+ P_{uav} + 2NP_{user} }\\
f_{n,2}& = \frac{ \log _{2}\left(1+\frac{p_{n,2} |{h}_{n,2}|^2}{1 + p_{n,1} |{h}_{n,2}|^2 }\right)} { p_{n,2}+ P_{uav} + 2NP_{user} }.
\end{align}

Consider that all the CC and CE users want to increase their own EE.
A penalty must be added to prevent the users from consuming excessive transmit power.
The penalty function should be an increasing function of the power to disincentivize the potential selfish behaviors of the users.
Additionally, the interference for a user strictly increases if too much power transmit power is allocated to the other user in a user pair.
An exponential function is selected as the penalty function.
The utility functions can be rewritten as \vspace{-0.5em}
\begin{align}
f_{n,1} &= \frac{ \log _{2}\left(1+p_{n,1} |{h}_{n,1}|^2 \right) }{p_{n,1}+ P_{uav} + 2NP_{user} }-e^{\kappa p_{n,1}}\\
f_{n,2}& = \frac{ \log _{2}\left(1+\frac{p_{n,2} |{h}_{n,2}|^2}{ 1 + p_{n,1} |{h}_{n,2}|^2 }\right)} {p_{n,2}+ P_{uav} + 2NP_{user} }-e^{\kappa p_{n,2}},
\end{align}
where the penalty coefficient $\kappa$ depends on the transmit power of the UAV, the number of users, and the circuit power consumption.
We set the penalty coefficient $\kappa$ to control the penalty function value.
The optimal value of $\kappa$ can be obtained by one-dimensional search.
When the other dependent parameters of $\kappa$ are fixed, the whole system attains different EE values by changing the value of parameter $\kappa$ in a communication cycle.
Then, we compare the EE values and select the parameter value with the maximum EE as the final $\kappa$.

Given $p_{n,k}$, the best response in terms of the $k$-th user's power is \vspace{-0.5em}
\begin{equation}\label{eq26}
{{p}_{n,k}}=\arg \underset{{{p}_{n,k}}\in \ {\mathcal{P}_{n,k}}}{\mathop{\max }}\,{{f}_{n,k}}\left( {{p}_{n,k}}\left| {p_{n,-k}} \right. \right),
\end{equation}
where $p_{n,-k}$ is the power strategy of all users, except for the $k$-th user of the $n$-th user pair.

\newtheorem{theorem}{Theorem}
\begin{theorem}
The utility function $f_{n,k}$ is a concave function if $p_{n,-k}$ is fixed.
\end{theorem}
\begin{proof}
See Appendix A.
\end{proof}

We can obtain the optimal solution to the transmit power allocation strategy response by using the interior point method and sequential quadratic programming.
The strategy of a user is iteratively adjusted according to those of the other users, to improve their utility.
The users' strategies can stabilize at a Nash equilibrium.

\begin{define}
A power strategy $p_n^* = ({p}_{n,k}^*,{p}_{n,-k}^*)$ is a Nash equilibrium of $\mathcal{G}_n$ if
\begin{equation}
\begin{aligned}
&f_{n,k}\left(p_{n,k}^{*}, p_{n,-k}^{*}\right) \geq f_{n,k}\left(p_{n,k}, p_{n,-k}\right), \\
&\ \ \ \ \ \ \ \ \forall p_{n,k}\in \mathcal{P}_{n,k}, n \in \mathcal{N}, k\in\{1,2\}.
\end{aligned}
\end{equation}
\end{define}

The existence of a Nash equilibrium is proven by using a supermodular game, as follows.

\begin{define} The game $\mathcal{G}_n = [{\mathcal{U}}_n, {\mathcal{P}}_n, {\mathcal{F}}_n]$ is a supermodular game if it satisfies the following conditions:
\begin{flushleft}1) Given $ p_{n,-k}$, the strategy space $\mathcal{P}_{n,i}$ is a nonempty and compact sublattice.\end{flushleft}
2) The utility function $f_{n,k}$ is twice differentiable in $p_{n,k}$ and $\frac{{\partial}^2 f_{n,k}}{\partial p_{n,k} \partial p_{n,k}'}\geq0$, where $p_{n,k}'$ is the power strategy of any user except for the $k$-th user of the $n$-th user pair.
\end{define}

According to Definition 2, we show that the proposed game is a supermodular game when the following
conditions are satisfied.

\begin{theorem}
The game $\mathcal{G}_n = [{\mathcal{U}}_n, {\mathcal{P}}_n, {\mathcal{F}}_n]$ is a supermodular game if and only if
\begin{equation}
\label{eq27}
 p_{n,2} \geq \sqrt{\frac{({P_{uav}+2NP_{user} })\left(p_{n,1} {{\left| {{h}_{n,2}} \right|}^{2}} + 1 \right)}{ {{\left| {{h}_{n,2}} \right|}^{2}}} }
\end{equation}
under the condition of $\left|{h}_{n,1}\right|>\left|{h}_{n,2}\right|$.
\end{theorem}
\begin{proof}
See Appendix B.
\end{proof}

The supermodular game $\mathcal{G}_n = [{\mathcal{U}}_n, {\mathcal{P}}_n, {\mathcal{F}}_n]$ has the following important properties.
\begin{property}
The following properties hold for the supermodular game $\mathcal{G}_n = [{\mathcal{U}}_n, {\mathcal{P}}_n, {\mathcal{F}}_n]$.

1) There exists a least one Nash equilibrium. The largest and smallest Nash equilibria
always exist, where the ``largest'' and ``smallest'' refer to the best responses of the power allocation coefficient strategy.

2) The strategies monotonically converge to the smallest
(largest) Nash equilibrium if the users' best responses are single-valued and each user updates its power from the smallest (largest) element of the strategy space.
\end{property}

Based on Theorem 2 and Property 1, we can infer that the supermodular game guarantees the existence of a Nash equilibrium.
Therefore, we design Algorithm 1 to obtain the Nash equilibrium by setting
adequate initialization conditions.
The strategy space ${\mathcal{P}}_{n,2}= \{p_{n,2}\mid p_{n,1}+p_{n,2}\leq p_{n}, p_{n,2} \geq \sqrt{\frac{({P_{uav}+2NP_{user} })\left(p_{n,1} {{\left| {{h}_{n,2}} \right|}^{2}} + 1 \right)}{ {{\left| {{h}_{n,2}} \right|}^{2}}} } \}$ can be adjusted until there is a Nash equilibrium. However, if $p_{n,2}$ is larger than a certain fixed value, this makes the strategy space ${\mathcal{P}}_{n,2}$ empty, which contradicts Definition 2.
Hence, we restrict the range of $p_{n,1}$ and
adjust its strategy space as ${\mathcal{P}}_{n,1}= \{p_{n,1}\mid p_{n,1}+p_{n,2}\leq p_{n}, 0 \leq p_{n,1} \leq \frac{p_{n,2}^2}{P_{uav}+2NP_{user} }-\frac{ 1 }{{ {\left| {{h}_{n,2}} \right|}^{2}} } \}$.

Algorithm 1 is designed to iteratively attain the Nash equilibrium by meticulously specifying initial conditions.
Despite the utility function ${{f}_{n,1}}$ is not a function of ${{p}_{n,2}}$, the strategy space of User$_{n,1}$ is related to the User$_{n,2}$'s power strategy.
In each iteration, the player User$_{n,1}$ first obtains the power ${{p}_{n,1}}$ according to (\ref{eq26}) based on original allocated power;
next, the BS updates the power allocation strategy space ${{\mathcal{P}}_{n,2}}$;
then, the player User$_{n,2}$ obtains the power ${{p}_{n,2}}$ according to (\ref{eq26});
at the end, the BS updates the power allocation strategy space ${{\mathcal{P}}_{n,1}}$, completes the current round, and proceeds with the next round.
In this game model ${{\mathcal{G}}_{n}}$, a system state updating mechanism is considered, which enables player User$_{n,1}$ (User$_{n,2}$) to dynamically adjust the transmit power and utility function iteratively by observing the outcome of the strategies of the other player  User$_{n,2}$ (User$_{n,1}$) until no player has any motivation to unilaterally change its strategy.

Algorithm 1 iterates and finally converges to a Nash equilibrium, where the power strategies of the users with weak channel conditions are limited by the lower bounds (\ref{eq27}) and the power strategies of the users with good channel conditions are limited by the upper bounds in $\mathcal{P}_{n,1}$. The user's fairness required by the SIC process can be provided in our system. In this way, we can effectively obtain rational power allocation strategies and power allocation coefficients.

\begin{algorithm}[t]
\caption{ Supermodular Game-based Nash Equilibrium for Optimization of the Power Allocation Coefficients of CC and CE Users.}
\begin{algorithmic}[1]
\STATE \textbf{Input}: \\
 \ \ 1) Initialize the maximum power $p_n^{max}$ and the maximum iteration index $I$; \\
 \ \ 2) Initialize the variables $p_{n,1}[0] = p_n^{max}/2$, $p_{n,2}[0]= \sqrt{\frac{\left(P_{u a v}+2 N P_{u s e r}\right)\left(p_{n, 1}\left|h_{n, 2}\right|^{2}+ 1 \right)}{\left|h_{n, 2}\right|^{2}}}.$
\FOR {$i = 0$ to $I$}
\STATE  {Update the power strategy $p_{n,1}$ by solving

${{p}_{n,1}[i+1]}=\arg \underset{{{p}_{n,1}}\in \ {\mathcal{P}_{n,1}}}{\mathop{\max }}\,{{f}_{n,1}}\left( {{p}_{n,1}}\left| {p_{n,-1}} \right. \right)$.
}
\STATE  {Update the power strategy space $\mathcal{P}_{n,2}$:

${\mathcal{P}}_{n,2}= \{p_{n,2}\mid p_{n,1}[i+1]+p_{n,2}\leq p_{n}, p_{n,2} \geq \sqrt{\frac{({P_{uav}+2NP_{user} })\left(p_{n,1}[i+1] {{\left| {{h}_{n,2}} \right|}^{2}} + 1 \right)}{ {{\left| {{h}_{n,2}} \right|}^{2}}} } \}$.
}
\STATE  {Update the power strategy $p_{n,2}$ by solving

${{p}_{n,2}[i+1]}=\arg \underset{{{p}_{n,2}}\in \ {\mathcal{P}_{n,2}}}{\mathop{\max }}\,{{f}_{n,2}}\left( {{p}_{n,2}}\left| {p_{n,-2}} \right. \right)$.
}
\STATE  {Update the power strategy space $\mathcal{P}_{n,1}$:

${\mathcal{P}}_{n,1}= \{p_{n,1}\mid p_{n,1}+p_{n,2}[i+1]\leq p_{n}, 0 \leq p_{n,1} \leq \frac{p_{n,2}[i+1]^2}{P_{uav}+2NP_{user} }-\frac{ 1 }{{ {\left| {{h}_{n,2}} \right|}^{2}} } \}$.
}
 \IF {$|{p_{n,1}[i+1]-p_{n,1}[i]|<\varepsilon }$}
  \STATE Break; \
  \ELSE
  \STATE Set $i= i+1$; \
 \ENDIF
\ENDFOR
\STATE  Calculate the power allocation coefficients of CC and CE users, $\alpha_{n,1}=\frac{p_{n,1}}{p_n}$ and  $\alpha_{n,2}=\frac{p_{n,2}}{p_n}$.
\STATE  Calculate the sum of utility functions, $f_{n,1}$ and $f_{n,2}$.
\STATE \textbf{Output}: Power allocation coefficients of CC and CE users, $\alpha_{n,1}$ and $\alpha_{n,2}$.
\label{code:recentEnd}
\end{algorithmic}
\end{algorithm}

\subsection{WPT Time Allocation, and Power Allocation between User Pairs}
Given the fixed power allocation coefficients $\alpha_{n,1}$ and $\alpha_{n,2}$ ($n\in \mathbb{N}$) and UAV position ($x_0$, $y_0$, $h$), problem $ \mathbf{P 1}$ in (\ref{equP1}) can be rewritten as
\begin{subequations}
\label{equP3}
\begin{align}
& \ \ \ \ \ \ \mathbf{P 3}:\ \ \underset{\{{p_n}\}_{n=1}^{N}, {\tau} }{\mathop{\max }} \ \ \Xi_{EE} \\
& s.t.\ \ \mathbf{C1}:\ \ \ 0\le \sum\nolimits_{n=1}^{N}{{p}_{n}}\le {{P}_{T}}, \ \  \forall n \in \mathcal{N} \\
&\ \ \ \ \ \ \mathbf{C2}:\ \ \ 0\le \tau \le T  \\
&\ \ \ \ \ \ \mathbf{C3}:\ \ \ R_{n,1}^{CC} \geq R_{\min} \\
&\ \ \ \ \ \ \mathbf{C4}:\ \ \ R_{n,2}^{CE} \geq R_{\min}\\
&\ \ \ \ \ \ \mathbf{C8}:\ \ \ \sum_{n=1}^{N} p_{n} \mathbf{q}_{n}^{H} \mathbf{q}_{n} \leq P_{\iota}, \iota=1, \ldots N.
\end{align}
\end{subequations}

Problem $\mathbf{P 3}$ is nonconvex and difficult to solve due to the nonconvexity of its objective (28a).
We first transform the objective function, i.e., (28a), as given by
\begin{small}
\begin{equation}
\label{equPT3}
\begin{aligned}
&\sum_{i=1}^{N}\left( \frac{ (T-\tau)B\log_{2}\left(1+\alpha_{n,1} p_{n} |{h}_{n,1}|^2 \right) }{ E_{sum} } \right.\\
& \ \ \ \ \ \left. +\frac{ (T-\tau) B \log _{2}\left(1+\frac{\alpha_{n,2} p_{n} |{h}_{n,2}|^2}{ 1 +\alpha_{n,1} p_{n} |{h}_{n,2}|^2 }\right)} { E_{sum} }\right) \\
& = \sum_{n=1}^{N} \frac{ (T-\tau)B\log_{2}\left(1+\alpha_{n,1} p_{n} |{h}_{n,1}|^2 \right) }{ E_{sum} }\\
&+\sum_{n=1}^{N}\frac{ (T-\tau) B \log _{2}\left({ 1 +\alpha_{n,1} p_{n} |{h}_{n,2}|^2 }+{\alpha_{n,2} p_{n} |{h}_{n,2}|^2} \right)} { E_{sum} }\\
&-\sum_{n=1}^{N}\frac{ (T-\tau) B \log _{2}\left({ 1 +\alpha_{n,1} p_{n} |{h}_{n,2}|^2 } \right)} { E_{sum} }.
\end{aligned}
\end{equation}
\end{small} \vspace{-1.5em}

The three terms on the right-hand side (RHS) of (\ref{equPT3}) are concave,
according to the properties of the logarithmic functions.
We employ the SCA to convexify the three concave terms.
Given the inequality $-\ln(y)\geq -\ln(x)-\frac{1}{x}(y-x)$  $\forall x, y > 0$,
we provide the lower bound approximation of $ -\log_{2}\left({ 1+\alpha_{n,1} p_{n} |{h}_{n,2}|^2 } \right) $, as follows:
\begin{equation}
\begin{aligned}
& \ \ \ - { \log _{2}\left({ 1 +\alpha_{n,1} p_{n} |{h}_{n,2}|^2 } \right)}\\
&\geq - { \log _{2}\left({ 1 +\alpha_{n,1}p_{n}[j] |{h}_{n,2}|^2 } \right)} \\
& \ \ \ - { \frac{\alpha_{n,1} |{h}_{n,2}|^2}{\ln2 \left({  1 +\alpha_{n,1}p_{n}[j] |{h}_{n,2}|^2 } \right)}}\left(p_{n}-p_{n}[j]\right).
\end{aligned}
\end{equation}
Further, $-\sum_{n=1}^{N} \frac{ (T-\tau) B \log _{2}\left({ 1+\alpha_{n,1} p_{n} |{h}_{n,2}|^2 } \right)}{E_{sum}} $ in (\ref{equPT3})
can be approximated by
\begin{equation}
\begin{aligned}
  & -\sum_{n=1}^{N} \frac{ (T-\tau) B \log _{2}\left({ 1 +\alpha_{n,1} p_{n} |{h}_{n,2}|^2 } \right)} {E_{sum}}\\
  & \geq -\sum_{n=1}^{N} \frac{ (T-\tau[j]) B \log _{2}\left({ 1+\alpha_{n,1}p_{n}[j] |{h}_{n,2}|^2 } \right)} {E_{sum}} \\
  & -\sum_{n=1}^{N}    \frac{(T-\tau[j]) B \alpha_{n,1} |{h}_{n,2}|^2}{\ln2 \left({ 1+\alpha_{n,1}p_{n}[j] |{h}_{n,2}|^2 } \right)} \frac{\left(p_{n}-p_{n}[j]\right) }{E_{sum}}.
\end{aligned}
\end{equation}
As a result, problem $\mathbf{P 3}$ can be converted to a convex problem.
We can invoke sequential quadratic programming to
solve the problem by rewriting it in the following minimal form:
\begin{subequations}
\label{equPT4}
\begin{align}
& \mathbf{P 4}: \underset{\{{p_n}\}_{n=1}^{N}, {\tau} }{\mathop{\min }} \{ -\sum_{n=1}^{N} \frac{ (T-\tau)B\log_{2}\left(1+\alpha_{n,1} p_{n} |{h}_{n,1}|^2 \right) }{ E_{sum} } \nonumber\\
& -\sum_{n=1}^{N}\frac{ (T-\tau) B \log _{2}\left({ 1+\alpha_{n,1} p_{n} |{h}_{n,2}|^2 }+{ \alpha_{n,2} p_{n} |{h}_{n,2}|^2} \right)} { E_{sum} }\nonumber\\
&+\sum_{n=1}^{N} \frac{ (T-\tau[j]) B \log _{2}\left({ 1 + \alpha_{n,1}p_{n}[j] |{h}_{n,2}|^2 } \right)} {E_{sum}} \nonumber \\
& +\sum_{n=1}^{N}    \frac{(T-\tau[j]) B \alpha_{n,1} |{h}_{n,2}|^2}{\ln2 \left({ 1 +\alpha_{n,1}p_{n}[j] |{h}_{n,2}|^2 } \right)} \frac{\left(p_{n}-p_{n}[j]\right) }{E_{sum}}  \} \\
& s.t. \ \ \ \ \ \mathbf{C1}\sim \mathbf{C4}, \mathbf{C8}.
\end{align}
\end{subequations}

Algorithm 2 summarizes the steps to solve problem $\mathbf{P 3}$, where we initialize the transmit power and WPT time variables $p_{n}[0]$ and $\tau_{n}[0]$, and then solve problem $\mathbf{P 4}$ in (\ref{equPT4}).
The solution to problem $\mathbf{P 4}$ in the $j$-th iteration, i.e., $p_{n}[j]$ and $\tau[j]$, is used as the local point of the next $(j+1)$-th iteration of the SCA.
This process repeats until the objective function converges.

\begin{algorithm}[t]
\caption{ SCA-Based Power and WPT Time Allocation for Each User Pair}
\begin{algorithmic}[1]
\STATE \textbf{Input}: \\
 \ Initialize the variables $p_{n}[0] \geq 0$ and $\tau[0] \geq 0$; \\
\ Initialize the iteration number $j=0$ and maximum iteration index $J$.
\WHILE {$|F(p_n[j+1],\tau[j+1])-F(p_n[j],\tau[j])|>\ell$ or $j>J$}
\STATE  {Obtain the optimal power $p_{n}[j+1]$  and optimal time allocation $\tau_{n}[j+1]$ by solving problem $\mathbf{P 4}$ in (\ref{equPT4})}.
\STATE  {Calculate the objective function value $F(p_n[j+1],\tau[j+1])$}.
\STATE  {$j\leftarrow j+1$}
\ENDWHILE
\STATE \textbf{Output}: Optimal power allocation $p_n^*$ and WPT time allocation $\tau^*$.
\label{code:recentEnd}
\end{algorithmic}
\end{algorithm}
\vspace{-1em}
\subsection{UAV Placement}
Given fixed optimal transmit power allocation, power allocation coefficients, and WPT time allocation, $\mathbf{P 1}$ can be transformed into a convex problem of UAV placement, as given by
\begin{subequations}
\label{equP5}
\begin{align}
& \mathbf{P 5}:\ \ \underset{ {x_0},{y_0} }{\mathop{\max }} \ \ \Xi_{EE} \\
& s.t.
\ \ \ \ \ \ \mathbf{C5 },\mathbf{C6}.
\end{align}
\end{subequations}
Problem $\mathbf{P 5}$ is convex and can be solved by using the Lagrange multiplier method \cite{Boyd2004Convex}.
The Lagrangian function is given by
\begin{equation}
\begin{aligned}
&\mathcal{L}={{\Xi }_{EE}}+{{\gamma }_{1}}\left( {{x}_{\min }}-{{x}_{0}} \right)+{{\gamma }_{2}}\left( {{x}_{0}}-{{x}_{\max }} \right)\\
& \ \ \ \ \ +{{\gamma }_{3}}\left( {{y}_{\min }}-{{y}_{0}} \right)+{{\gamma }_{4}}\left( {{y}_{0}}-{{y}_{\max }} \right),
\end{aligned}
\end{equation}
where ${{\gamma }_{i}},\forall i \in [1,4]$ gives the Lagrange multipliers.
The optimal UAV location can be obtained by applying the
Karush-Kuhn-Tucker (KKT) conditions:
\begin{align}
\frac{\partial \mathcal{L} }{\partial x_0} = 0,
\frac{\partial \mathcal{L} }{\partial y_0} = 0, \ \ \ \ \ \ \ \ \ \ \ \ \ \ \ \ \ \ \ \ \ \ \ \ \nonumber \\
\ x_{\min } \leq x_{0} \leq x_{\max },  
\ y_{\min } \leq y_{0} \leq y_{\max }, \nonumber \\
{\gamma}_{i}\geq 0,  \forall i \in [1,4], \ \ \ \ \ \ \ \ \ \ \ \ \ \ \ \ \ \ \ \ \ \ \ \ \ \nonumber\\
{{\gamma }_{1}}\left( {{x}_{\min }}-{{x}_{0}} \right)={{\gamma }_{2}}\left( {{x}_{0}}-{{x}_{\max }} \right)= \ \ \ \ \nonumber \\
{{\gamma }_{3}}\left( {{y}_{\min }}-{{y}_{0}} \right)={{\gamma }_{4}}\left( {{y}_{0}}-{{y}_{\max }} \right)=0. \ \
\end{align}
The position of the UAV satisfies $x_{\min} < x_0 <
x_{\max}$ and $y_{\min} < y_0 < y_{\max}$, indicating ${\gamma}_{i} = 0$.
Here, $\frac{\partial \mathcal{L} }{\partial x_0}$ and $\frac{\partial \mathcal{L} }{\partial y_0}$ are given in (\ref{eqn_long1}) and (\ref{eqn_long2}), where, for notational simplicity, we define
${{\mathcal{Q}}_{1}}={M{\alpha}_{n,1}}\beta{\beta_{0}}{{p}_{n}} |\mathbf{a}\left(\theta_{n,k}\right)|^2 |\mathbf{q}_{n,1}|^{2}$,
${{\mathcal{Q}}_{2}}={M{\alpha}_{n,2}}\beta{\beta_{0}}{{p}_{n}} |\mathbf{a}\left(\theta_{n,k}\right)|^2 |\mathbf{q}_{n,2}|^{2}$,
${{\mathcal{R}}_{1}}=2{{x}_{0}}-2{{x}_{n,2}}$,
${{\tilde{\mathcal{R}}}_{1}}=2{{y}_{0}}-2{{y}_{n,2}}$,
${{\mathcal{R}}_{2}}=\frac{{{\mathcal{Q}}_{2}}}{{{\mathcal{R}}_{6}}{{\mathcal{R}}_{9}}^{\beta }}+1$,
${{\mathcal{R}}_{3}}=\tau \left( P+M{{P}_{a}} \right)+T\left[ {{P}_{0}}+{{P}_{1}}+M{{P}_{a}}+\frac{\mathbf{H}_{PU}^{2}P\xi \tau }{\left( T-\tau  \right){{\mathcal{R}}_{7}}^{\beta }} \right]+2NT{{P}_{user}}$,
${{\mathcal{R}}_{4}}=\frac{{{\mathcal{Q}}_{1}}}{{{\mathcal{R}}_{8}}^{\beta }}+1$,
${{\mathcal{R}}_{5}}={{\mathcal{R}}_{9}}^{\beta +1}$,
${{\mathcal{R}}_{6}}=1+\frac{  \mathcal{Q}_1 |\mathbf{q}_{n,2}|^{2} } {|\mathbf{q}_{n,1}|^{2}{{\mathcal{R}}_{9}}^{\beta }}$,
${{\mathcal{R}}_{7}}=\sqrt{{{\left( {{x}_{0}}-{{x}_{b}} \right)}^{2}}+{{\left( {{y}_{0}}-{{y}_{b}} \right)}^{2}}+{{h}^{2}}}$,
${{\mathcal{R}}_{8}}=\sqrt{{{\left( {{x}_{0}}-{{x}_{n,1}} \right)}^{2}}+{{\left( {{y}_{0}}-{{y}_{n,1}} \right)}^{2}}+{{h}^{2}}}$,
${{\mathcal{R}}_{9}}=\sqrt{{{\left( {{x}_{0}}-{{x}_{n,2}} \right)}^{2}}+{{\left( {{y}_{0}}-{{y}_{n,2}} \right)}^{2}}+{{h}^{2}}}$, and
${{\mathcal{R}}_{10}}=\mathbf{H}_{PU}^{2}NPT\beta\xi \tau$.

\newcounter{mytempeqncnt}
\begin{figure*}[!t]
\normalsize
\setcounter{mytempeqncnt}{\value{equation}}
\setcounter{equation}{36}
\begin{small}
\begin{equation}
\label{eqn_long1}
\frac{\partial \mathcal{L} }{\partial x_0} = \frac{ {{\mathcal{R}}_{10}} \left( 2{{x}_{0}}-2{{x}_{b}} \right)\left[ \frac{B\log ({{\mathcal{R}}_{2}}{{\mathcal{R}}_{4}})}{\log (2)} \right]}{2\mathcal{R}_{7}^{\beta +2}\mathcal{R}_{3}^{2}}-\frac{N\left( T-\tau \right)\left[ \frac{B\left( \frac{{{\mathcal{Q}}_{2}} {{\mathcal{R}}_{1}}}{2{{\mathcal{R}}_{6}}{{\mathcal{R}}_{5}}{{\mathcal{R}}_{9}}}
-\frac{{{\mathcal{Q}}_{1}}{{\mathcal{Q}}_{2}} {{\mathcal{R}}_{1}}|\mathbf{q}_{n,2}|^{2} }{2\beta|\mathbf{q}_{n,1}|^{2}\mathcal{R}_{6}^{2}{{\mathcal{R}}_{5}}\mathcal{R}_{9}^{\beta +1}} \right)}{\log (2){{\mathcal{R}}_{2}}}+\frac{B{{\mathcal{Q}}_{1}}\left( 2{{x}_{0}}-2{{x}_{n,1}} \right)}{2\log (2){{\mathcal{R}}_{4}}\mathcal{R}_{8}^{\beta +2}} \right]}{{{\mathcal{R}}_{3}}}-\gamma_1+\gamma_2=0
\end{equation}
\end{small}
\begin{small}
\begin{equation}
\label{eqn_long2}
\frac{\partial \mathcal{L} }{\partial y_0}= \frac{ {{\mathcal{R}}_{10}} \left( 2{{y}_{0}}-2{{y}_{b}} \right)\left[ \frac{B\log ({{\mathcal{R}}_{2}}{{\mathcal{R}}_{4}})}{\log (2)} \right]}{2\mathcal{R}_{7}^{\beta +2}\mathcal{R}_{3}^{2}}-\frac{N\left( T-\tau \right)\left[ \frac{B\left( \frac{{{\mathcal{Q}}_{2}} {{{\tilde{\mathcal{R}}}}_{1}}}{2{{\mathcal{R}}_{6}}{{\mathcal{R}}_{5}}{{\mathcal{R}}_{9}}}-\frac{{{\mathcal{Q}}_{1}}{{\mathcal{Q}}_{2}} {{{\tilde{\mathcal{R}}}}_{1}}|\mathbf{q}_{n,2}|^{2}}{2\beta|\mathbf{q}_{n,1}|^{2}\mathcal{R}_{6}^{2}{{\mathcal{R}}_{5}}\mathcal{R}_{9}^{\beta +1}} \right)}{\log (2){{\mathcal{R}}_{2}}}+\frac{B{{\mathcal{Q}}_{1}}\left( 2{{y}_{0}}-2{{y}_{n,1}} \right)}{2\log (2){{\mathcal{R}}_{4}}\mathcal{R}_{8}^{\beta +2}} \right]}{{{\mathcal{R}}_{3}}}-\gamma_3+\gamma_4 = 0.
\end{equation}
\end{small}
\setcounter{equation}{\value{mytempeqncnt}}
\hrulefill
\vspace*{4pt}
\end{figure*}
\setcounter{equation}{38}

The optimum $(x_0^*,y_0^*)$ can be obtained with the Lagrange multipliers.
We take the gradient method to update $\gamma_1$, $\gamma_2$, $\gamma_3$ and $\gamma_4$, as given by
\begin{align}
&\gamma_1(\omega+1)=\left[\gamma_1(\omega)-\psi_{\gamma_1}\left( x_{0}-x_{\min}\right)\right]^{+} \label{multi1} \\
&\gamma_2(\omega+1)=\left[\gamma_2(\omega)-\psi_{\gamma_2}\left( x_{\max}-x_{0}\right)\right]^{+} \label{multi2} \\
&\gamma_3(\omega+1)=\left[\gamma_3(\omega)-\psi_{\gamma_3}\left( y_{0}-y_{\min}\right)\right]^{+} \label{multi3}\\
&\gamma_4(\omega+1)=\left[\gamma_4(\omega)-\psi_{\gamma_4}\left( y_{\max}-y_{0}\right)\right]^{+} \label{multi4},
\end{align}
where $\psi_{\gamma_1}$, $\psi_{\gamma_2}$, $\psi_{\gamma_3}$ and $\psi_{\gamma_4}$ are the positive step sizes and $\omega$ is the iteration index.
Algorithm 3 summarizes the optimization of the UAV location.

\begin{algorithm}[t]
\caption{ Energy-Efficient UAV Location Designing Algorithm }
\begin{algorithmic}[1]
\STATE \textbf{Input}: \\
  Initialize the EE ${{\Xi }_{EE}^*}$, Lagrange multipliers $\gamma_1$, $\gamma_2$, $\gamma_3$ and $\gamma_4$, maximum number of iterations $N_{MAX}$ and maximum tolerance $\epsilon$. \\
\STATE {Obtain the UAV coordinates $x_{0}$ and $y_{0}$ in formula (\ref{eqn_long1}) and (\ref{eqn_long2}), and use the solution to calculate the EE ${{\Xi }_{EE}^*(\omega)}$. }
\STATE {Use (\ref{multi1}) $\sim$ (\ref{multi4}) to update $\gamma_1$, $\gamma_2$, $\gamma_3$ and $\gamma_4$ in the iteration procedure.}
\IF {$|{{\Xi }_{EE}^*(\omega)-{\Xi }_{EE}^*(\omega-1)|<\epsilon }$ or $\omega = N_{MAX}$}
  \STATE $x_0^* = x_0(\omega)$, $y_0^* = y_0(\omega)$ ; \
  \STATE  {$\omega\leftarrow \omega+1$};
  \ELSE
  \STATE Repeat step 2 and step 3; \
\ENDIF
\STATE \textbf{Output}: Optimal UAV location coordinates $x_0^*$ and $y_0^*$.
\label{code:recentEnd}
\end{algorithmic}
\end{algorithm}
\vspace{-1em}
\subsection{ Convergence and Complexity of Overall Algorithm }
As summarized in Algorithm 4, the overall algorithm solving problem $\mathbf{P 1}$ in (\ref{equP1}) consists of Algorithms 1, 2, and 3 by applying the BCD method \cite{Bertsekas}.
The optimization variables of the problem in (\ref{equP1}) are divided into three blocks, i.e., the intra-pair NOMA power allocation coefficient $\{\alpha_{n,k}\}$, the WPT time and inter-pair power allocation $\{{p_n},{\tau}\}$, and the UAV location $\{{x_0},{y_0}\}$.
We can obtain the solutions $\{\alpha_{n,k}\}$, $\{{p_n},{\tau}\}$, and $\{{x_0},{y_0}\}$ to problems $\mathbf{P 2}$, $\mathbf{P 3}$, and $\mathbf{P 5}$ by running Algorithms 1, 2, and 3, respectively.
The result of each iteration is taken as the input to the next iteration, until the convergence of the BCD.

In the classical BCD method, convergence is ensured when each subproblem is solved optimally in each iteration.
In Algorithm 4, the WPT time and inter-pair power allocation subproblem is solved approximately in (\ref{equPT4}) by using SCA.
The convergence analysis of the classical BCD method is not directly applicable. Nevertheless, we can assert the convergence of Algorithm 4, as proved in Appendix C.

The computational complexity of the proposed algorithms is analyzed as follows.
Let $\varepsilon$, $\ell $, $\epsilon $ and $\varpi $ denote the convergence accuracies of Algorithms 1, 2, 3, and 4, respectively.
The complexity of Algorithm 1 is dominated by Steps 2 -- 8, which involves $\left( 95M+233 \right)$ additions and multiplications per iteration; therefore, the complexity is $\mathcal{O}\left( M\log \left( {1}/{\varepsilon }\; \right) \right)$.
The complexity of Algorithm 2 is dominated by Step 3, which involves $\left( {{M}^{2}}+38MN+99N+19 \right)$ additions and multiplications per iteration; therefore, the complexity is $\mathcal{O}\left( \left( {{M}^{2}}+MN \right)\log \left( {1}/{\ell }\; \right) \right)$.
The complexity of Algorithm 3 is dominated by Step 2, which involves $\left( 2{{M}^{2}}+12MN+15N+84\beta +622 \right)$ additions and multiplications per iteration; therefore, the complexity is $\mathcal{O}\left( \left( {{M}^{2}}+MN \right)\log \left( {1}/{\epsilon }\; \right) \right)$.
Since these three algorithms are executed sequentially in Algorithm 4, the overall computational complexity of Algorithm 4 is $\mathcal{O}\left( \left( M\log \left( {1}/{\varepsilon }\; \right)+\left( {{M}^{2}}+MN \right)\log \left( {1}/{\left( \ell \epsilon  \right)}\; \right) \right)\log \left( {1}/{\varpi }\; \right) \right)$.

\begin{algorithm}[t]
\caption{Supermodular Game-Based Resource Allocation and UAV Location Design Algorithm for Problem (\ref{equP1})}
\begin{algorithmic}[1]
\STATE \textbf{Initialize}: \\
  Initialize $\{{p_n}^{(0)},{\tau}^{(0)}\}$, $\{{x_0}^{(0)},{y_0}^{(0)}\}$, and iteration index $r$. \\
\STATE \textbf{repeat}: \\
\STATE  According to the Algorithm 1, solve subproblem (\ref{equP2}) for given $\{{p_n}^{(r)},{\tau}^{(r)}\}$ and $\{{x_0}^{(r)},{y_0}^{(r)}\}$, and denote the optimal solution of the intrapair NOMA power allocation coefficient as $\{\alpha_{n,k}^{(r+1)}\}$.\\
\STATE  According to the Algorithm 2, solve subproblem (\ref{equPT4}) for given $\{\alpha_{n,k}^{(r+1)}\}$, and $\{{p_n}^{(r)},{\tau}^{(r)}\}$, $\{{x_0}^{(r)},{y_0}^{(r)}\}$, and denote the optimal solution of the WPT time and interpair power allocation as $\{{p_n}^{(r+1)},{\tau}^{(r+1)}\}$.\\
\STATE  According to the Algorithm 3, solve subproblem (\ref{equP5}) for given $\{\alpha_{n,k}^{(r+1)}\}$, and $\{{p_n}^{(r+1)},{\tau}^{(r+1)}\}$, $\{{x_0}^{(r)},{y_0}^{(r)}\}$, and denote the optimal solution of the  UAV location as $\{{x_0}^{(r+1)},{y_0}^{(r+1)}\}$.\\
\STATE \textbf{until} The fractional increase of the objective EE value is
below a threshold $\varpi >0$.
\label{algorithm4}
\end{algorithmic}
\end{algorithm}

\section{Simulation Results }
In this section, we provide simulation results of the proposed algorithm.
The positions of users are generated according to a uniform random distribution in a circular range with a radius of 50 m.
The minimum distance between the UAV and any user is 1 m.
Both path loss and Rayleigh fading are considered.
The parameters of the probabilistic LoS channel model are $\beta_{0} = 1.8$ and $\beta = 1.1$. The bandwidth $B$ is 1 Hz.
Considering realistic scenarios, the configuration of the multiantenna, rotary-wing UAV is specified in Table \uppercase\expandafter{\romannumeral1} with reference to \cite{8663615}, \cite{filippone2006flight} and \cite{7888557}. The other simulation parameters are listed in Table \uppercase\expandafter{\romannumeral2} with reference to \cite{8672190}, \cite{8488592} and \cite{9348068}.

\begin{table}[h]
	\centering{}
    \label{table1}
	\textbf{Table \uppercase\expandafter{\romannumeral1}}~~TECHNICAL CONFIGURATION OF A UAV.\\ \vspace{0.5em}
	\setlength{\tabcolsep}{2mm}{
		\begin{tabular}{ll|ll} \toprule
			Parameter   & Value  & Parameter  & Value \\
			\midrule
			$P_0$ & 0.1   & $ \rho $ &  1.225   \\
            $U_{tip}$   &  200   &  $s$   & 0.05 \\
			$P_i$     & 0.2    & $ A $ & 0.79 \\
            $v_0$   & 7.2  & $d_0$   & 0.3     \\
            \toprule
	\end{tabular}}
\end{table}
\begin{table}[h]
	\centering{}
    \label{table1}
	\textbf{Table \uppercase\expandafter{\romannumeral2}}~~SIMULATION PARAMETERS of SYSTEM OPTIMIZATION.\\ \vspace{0.5em}
	\setlength{\tabcolsep}{2mm}{
		\begin{tabular}{ll|ll} \toprule
			Parameter   & Value  & Parameter   & Value \\
			\midrule
			$N$     & 10$\sim $40     & $ \eta $ &  0.6 $\sim$ 0.9  \\
			$M$     & 4 $\sim $32   & $ \epsilon $ & $10^{-7}$  \\
            $\beta$    &  1.05    &  $ x_b $ & -60 m      \\
            $P$   & 30 $\sim$ 47 dBm  & $ y_b $ &  $0$ m  \\
            $\beta_0$   & 1.8 & $ N_\text{MAX} $ & 1000 \\
            $R_{\min}$   & 0.1 bit/s/Hz & $B$  & 1 Hz \\
            $h$ & $ 5 \sim 50$  m & $\kappa$   & 0.05\\
            $y_{\max}$   & 50 m & $y_{\min}$ & -50 m\\
            $x_{\min}$ & -50 m & $x_{\max}$ & 50 m  \\
            \toprule
	\end{tabular}}
\end{table}

Fig. 3 shows the EE performance for the wireless-powered UAV-assisted NOMA system under different numbers of user pairs ranging from 10 to 30.
We use ``SMG'', ``RA'' and ``LD'' to denote the supermodular game, resource allocation and location design, respectively.
We also consider the OMA (i.e., OFDMA) transmission scheme developed in \cite{8809094}.
As shown in Fig. 3, with the increasing number of user pairs, the system EE slowly increases and gradually stabilizes.
Fig. 3 also shows that the proposed Algorithm 4 (i.e., SMG-based RA \& LD) with energy harvesting (EH) outperforms that without EH and the OMA-based alternative.
For example, when $N$ = 25 and $p_n$ = 1 W for the $n$-th pair users ($\forall n$), the EE of our algorithm is 80$\%$ higher than that without EH and 25$\%$ higher than the OMA-based algorithm.
This is because with the power-domain multiplexing and EH, the UAV can efficiently utilize the available spectrum in the NOMA system and hence substantially improve the EE.
Additionally, we compare the proposed Algorithm 4 under $p_n$ = 1 W and $p_n$ = 2 W. The algorithm is better when $p_n$ = 2 W, since a reasonably higher power within the control range for each user pair can improve the data rate and subsequently the EE.

\begin{figure}[t]
\center
\includegraphics[width=8cm]{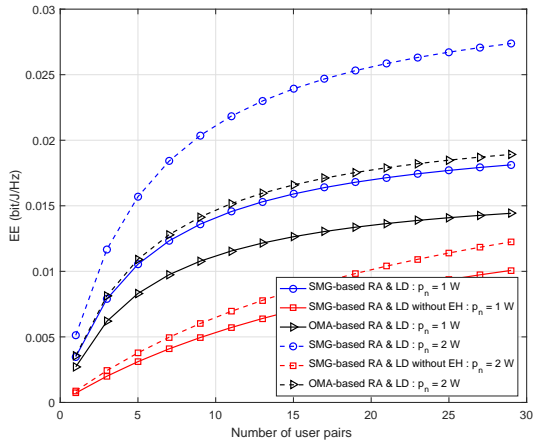}
\caption{EE comparison of the SMG-based RA and LD algorithm and OMA-based alternative.}
\end{figure}

Fig. 4 shows the EE performance of the proposed SMG-based RA and LD algorithm, i.e., Algorithm 4, versus the transmit power of the power beacon.
``ETPA'' stands for the existing transmit power allocation method developed in \cite{7982784}.
As the transmit power of the power beacon increases, the EE first increases, then reaches its peak, and then moderately declines.
We also compare Algorithm 4 with the ETPA-based RA and LD algorithm when the number of GUs is 20, 30 and 40.
We observe that the EE performance of Algorithm 4 is more energy-efficient than that of the ETPA-based algorithm.
Interestingly, under Algorithm 4, the EE is lower when there are 30 users than it is when there are 40 users, in the case where the transmit power of the power beacon is lower than 12 W; and the other way around in the other cases.
This finding indicates that the power beacon needs to transmit more RF energy when the number of user pairs increases to sustain the UAV's services and deliver quality of service (QoS).

\begin{figure}[t]
\center
\includegraphics[width=8cm]{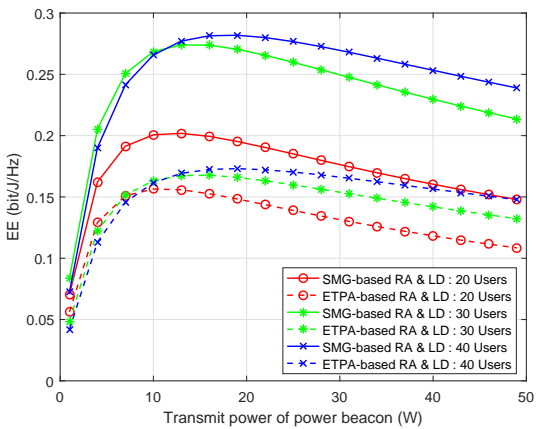}
\caption{Comparison of the system EE versus transmit power of the power beacon for the proposed SMG-based RA and LD and OMA-based RA and LD algorithms.}
\end{figure}

In Fig. 5, we compare the EE for the proposed Algorithm 4 with and without considering NOMA or UAV location design.
Algorithm 4 supporting both NOMA and LD outperforms that only supporting NOMA, the OMA-based RA and LD algorithm, and the OMA-based RA algorithm with no LD.
For example, when the UAV serves 20 users with a transmit power of 5 W, the EE of Algorithm 4 is 0.02266 bit/J/Hz, which is larger than those in the other three cases, i.e., 0.01456 bit/J/Hz, 0.008651 bit/J/Hz, and 0.005405 bit/J/Hz.
One reason for the gain of the proposed algorithm is that NOMA can improve the number of admitted GUs, thereby increasing the SE and EE.
Another reason is that a reasonable UAV location design can reduce the energy consumption resulting from path loss, and further improve the EE.
Additionally, the EE of Algorithm 4 is provided under different numbers of users.
Notably, when the transmit power of the UAV is small, the EE is higher under 20 users than it is under 30 users.
In contrast, when the transmit power is relatively large, the EE is lower under 20 users than it is under 30 users.
This is because, for more users, the low transmit power of the UAV is insufficient and cannot ensure a high system EE.

\begin{figure}
\center
\includegraphics[width=8cm]{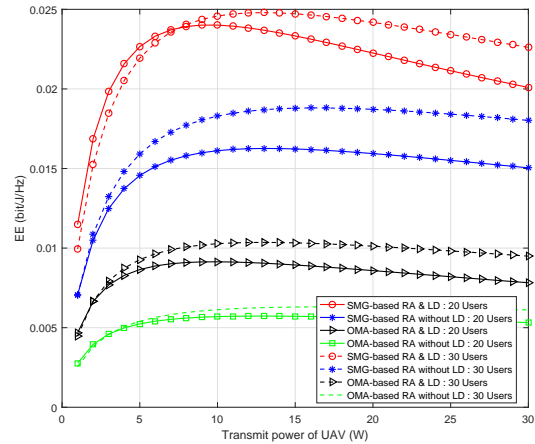}
\caption{Comparison of the EE versus transmit power of the UAV for the proposed SMG-based RA and LD algorithm under different transmission and location design schemes. }
\end{figure}

Fig. 6 shows the EE versus the power consumption of each user under the different optimization schemes and transmit powers of the UAV.
The EE of all the cases decreases, as the power consumption of each user increases.
This is because the more power is transmitted to a user, the higher total energy consumption the UAV-assisted NOMA system incurs;
thus, the total EE of the system decreases.
We also compare the proposed Algorithm 4, the ETPA-based RA and LD algorithm, and the traditional OMA-based RA algorithm under different transmit powers of the UAV.
The proposed algorithm is the best among the three algorithms, which shows the efficiency of the SMG-based intra-pair power allocation presented in Algorithm 1.
In addition, the EE of the proposed algorithm is better when $p_n$ = 2 W than it is when $p_n$ = 1 W or 1.5 W.
The results demonstrate that allocating relatively high transmit power to each user pair contributes to a high system EE, which is consistent with Figs. 3 and 5.
\begin{figure}
\center
\includegraphics[width=8cm]{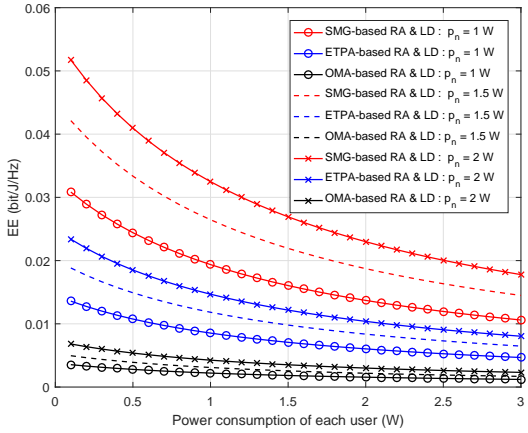}
\caption{Impact on EE of different power allocation approaches for different transmit powers.}
\end{figure}

Figs. 7(a) and 7(b) show the EE of the proposed SMG-based RA and LD algorithm and the ETPA-based RA and LD algorithm, i.e., Algorithm 4, under the different transmit antenna numbers. In Fig. 7(a), we observe that the EE of the proposed algorithm decreases, as the power consumption of each user increases. We also find that the multiantenna UAV outperforms its single-antenna counterpart.
Fig. 7(b) plots the EE of the proposed algorithm and its ETPA-based alternative as the transmit power of the power beacon increases, where the number of antennas is set to 1, 8, and 16 at the UAV.
More antennas result in better EEs under both algorithms. For example, when the transmit power of the power beacon is 20 W and the UAV is equipped with 16 antennas, the EE of the proposed  algorithm is 0.17 bit/J/Hz. The EE is higher than it is when the UAV is equipped with eight antennas (0.15 bit/J/Hz) or a single antenna (0.083 bit/J/Hz).

\begin{figure}
  \centering
  \subfigure[The EE of the proposed  algorithm.]{\includegraphics[width=3in]{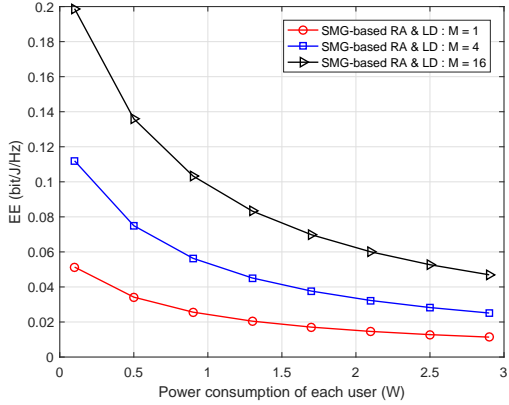}}
  \subfigure[Comparison of the considered algorithms.] {\includegraphics[width=3in]{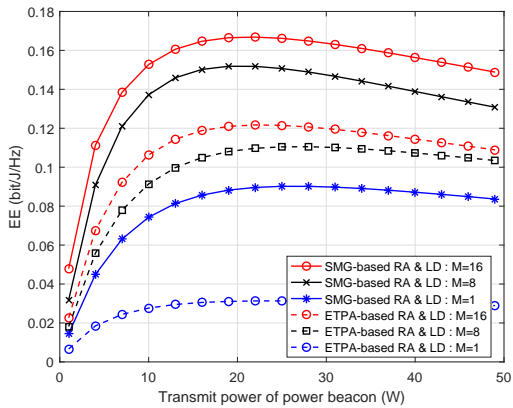}}
  \caption{The EE performance under different antenna numbers. }
\end{figure}

We also show the convergence performance of the proposed SMG-based RA and LD algorithm, the ETPA-based RA and LD algorithm, the OMA-based RA and LD algorithm, and the SMG-based RA algorithm without LD in Fig. 8.
We show that the system EE can reach the maximum with a small number of
iterations by adopting the proposed algorithm.
We also find that the proposed SMG-based RA and LD algorithm has the best EE performance compared to the other three algorithms.
When we employ the UAV LD scheme, the convergence rate is faster than that with no UAV LD scheme.
The reason is that we can choose the optimal UAV location to avoid excessive energy consumption.
\begin{figure}
\center
\includegraphics[width=8cm]{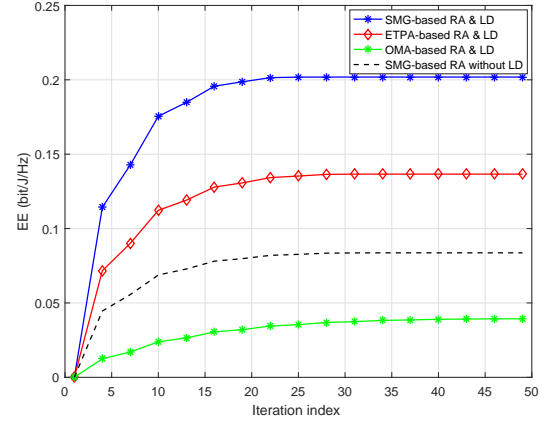}
\caption{Convergence of the proposed SMG-based RA and LD algorithm and reference algorithms.}
\end{figure}
\begin{figure}
\center
\includegraphics[width=8cm]{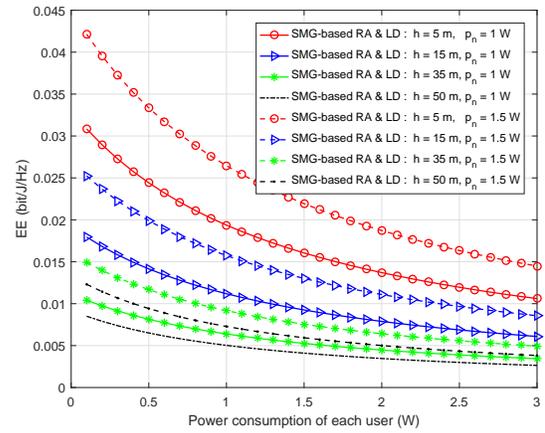}
\caption{EE versus power consumption of each user for the proposed SMG-based RA and LD algorithm under different heights and allocated transmit powers of the UAV. }\vspace{-1.5em}
\end{figure}

Fig. 9 shows that the EE changes differently with the power consumption of each user under different elevations and allocated transmit powers of the UAV. 
As the power consumption of each user increases, the system EE decreases, which is similar to Fig. 5.
Moreover, the system EE, which changes slightly under a high power consumption of each user, is almost entirely determined by the altitude of the UAV.
For example, when the altitude is 5 m, the EE is twice as large as it is at the altitude of 50 m. 
This is because the lower altitude the UAV is at, the shorter distance it has from the GU, resulting in less path loss.
However, there can be buildings, trees, or hilly terrains blocking transmission in practice.
A relatively high altitude of a UAV is recommended to ensure relatively high EE while maintaining LoS communication.
Moreover, allocating more transmit power to the $n$-th user pair can enhance the system EE.
At a given UAV altitude, the more power is allocated to the $n$-th user pair, the larger system EE is achieved.

Fig. 10 plots the total EE performance of the proposed SMG-based RA and LD algorithm versus the WPT time $\tau$. With the increasing WPT time, the EE first increases and then decreases.
We compare the EE under three cases of the transmit power, namely, Case I: $P = 20$ W, Case II: $P=30 $ W and Case III: $P=40 $ W.
We observe that, under different transmit powers of the power beacon, there is a different optimal EE.
The optimal energy harvesting time corresponding to the optimal EE is different.
This is because the larger transmit power can charge the UAV faster, thus achieving the optimal system EE faster in Case III than in the other two cases.
We also find that the maximum EE is higher in Case II than it is in Cases I and III, indicating that the energy loss of the RF chain caused by an excessive transmit power penalizes the system EE.

\begin{figure}
\center
\includegraphics[width=8cm]{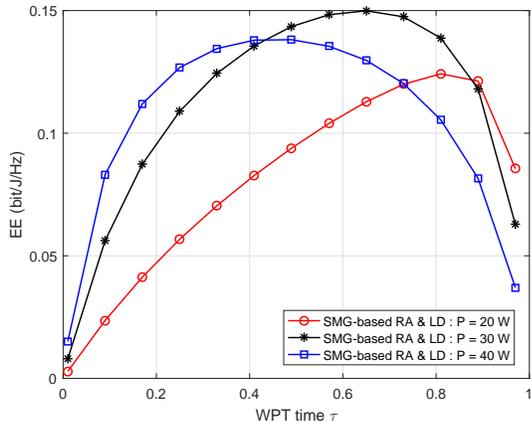}
\caption{The EE of the proposed SMG-based RA and LD algorithm under different WPT time. }
\end{figure}

Fig. 11 evaluates the EE performance of the proposed SMG-based RA and LD algorithm by optimizing the horizontal position of the UAV. We consider three cases with different $x$- and $y$-coordinates of the UAV. We observe that in each of the cases, there is an optimal UAV position accounting for the optimal EE. We also find that the optimal UAV position is close to the power beacon and a relatively large number of users, which demonstrates the effectiveness of the UAV placement.

\begin{figure*}[htbp]
\centering

\subfigure[Flat top view of 10 user systems.]{
\begin{minipage}[t]{0.33\linewidth}
\centering
\includegraphics[width=2in]{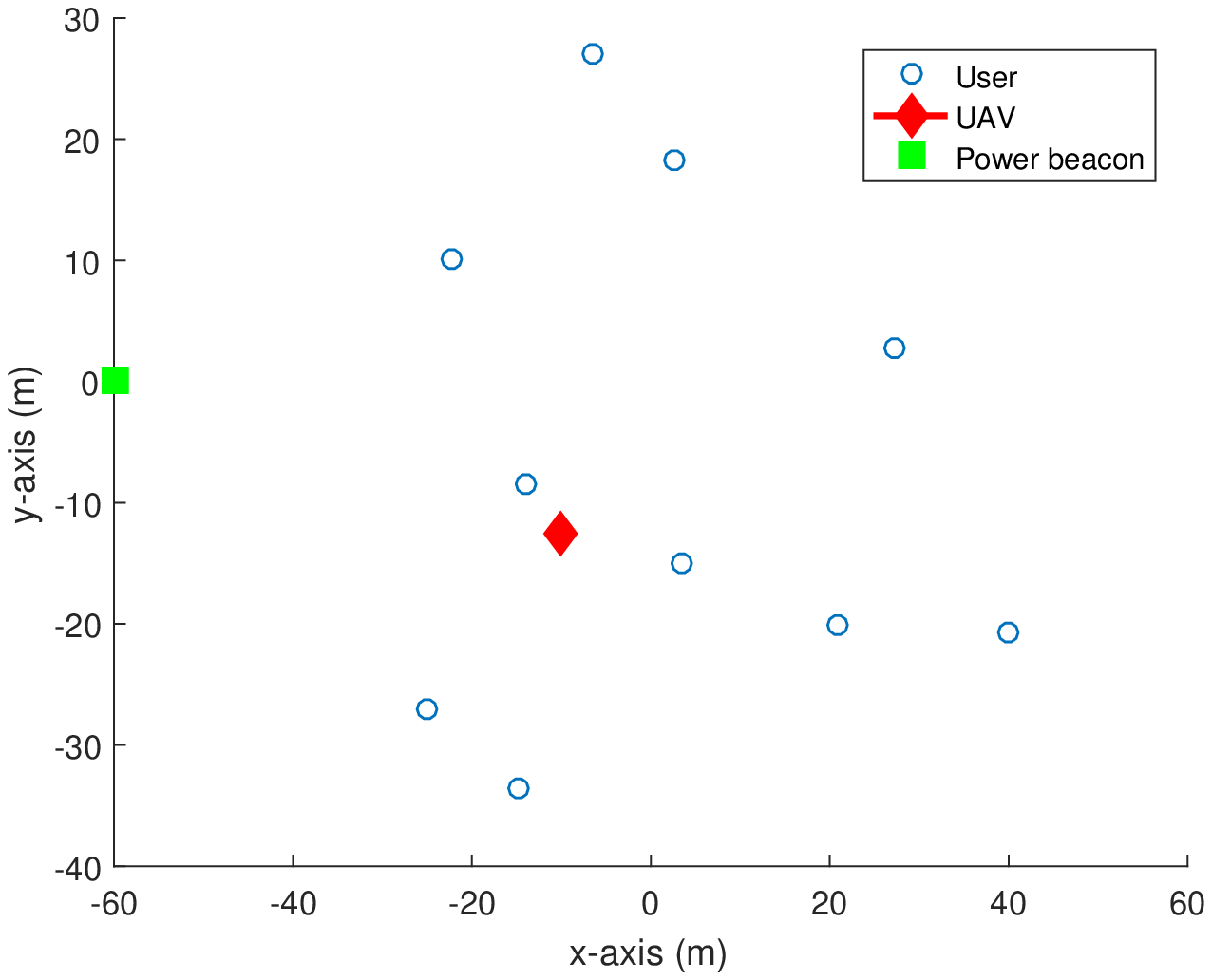}
\end{minipage}%
}%
\subfigure[Flat top view of 20 user systems.]{
\begin{minipage}[t]{0.33\linewidth}
\centering
\includegraphics[width=2in]{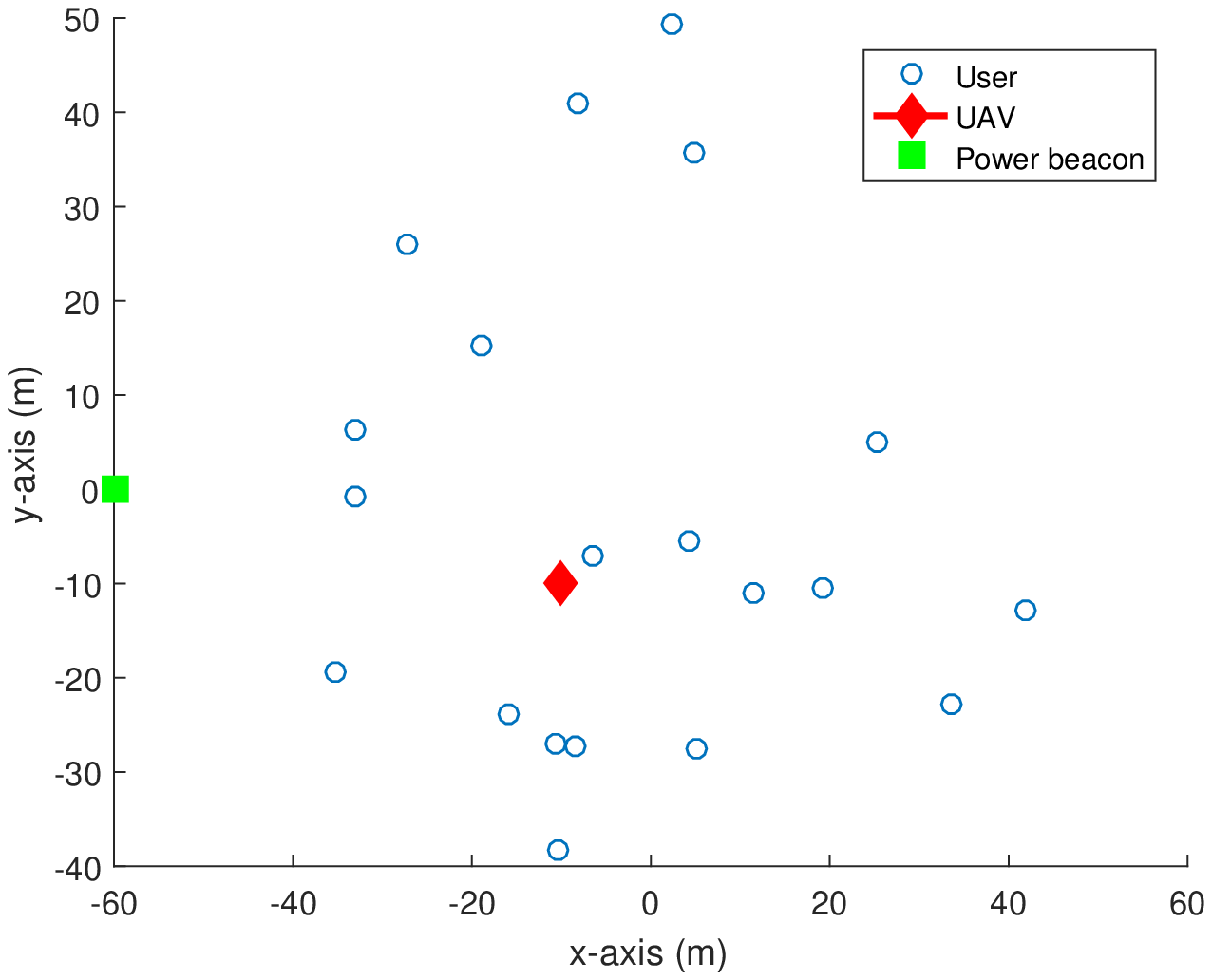}
\end{minipage}%
}%
\subfigure[Flat top view of 30 user systems.]{
\begin{minipage}[t]{0.33\linewidth}
\centering
\includegraphics[width=2in]{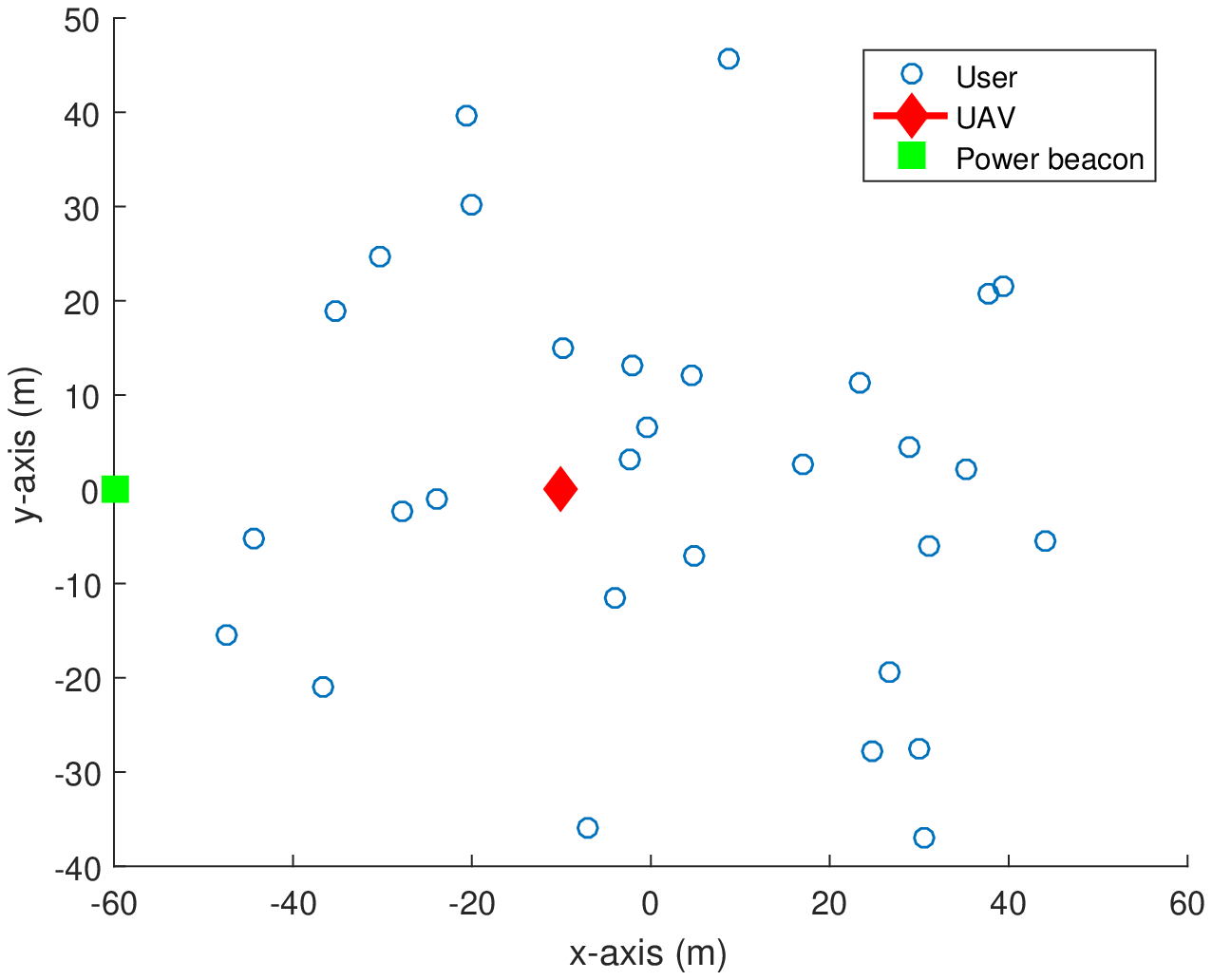}
\end{minipage}%
}%
\quad                
\subfigure[EE performance versus UAV location of 10 user systems.]{
\begin{minipage}[t]{0.33\linewidth}
\centering
\includegraphics[width=2in]{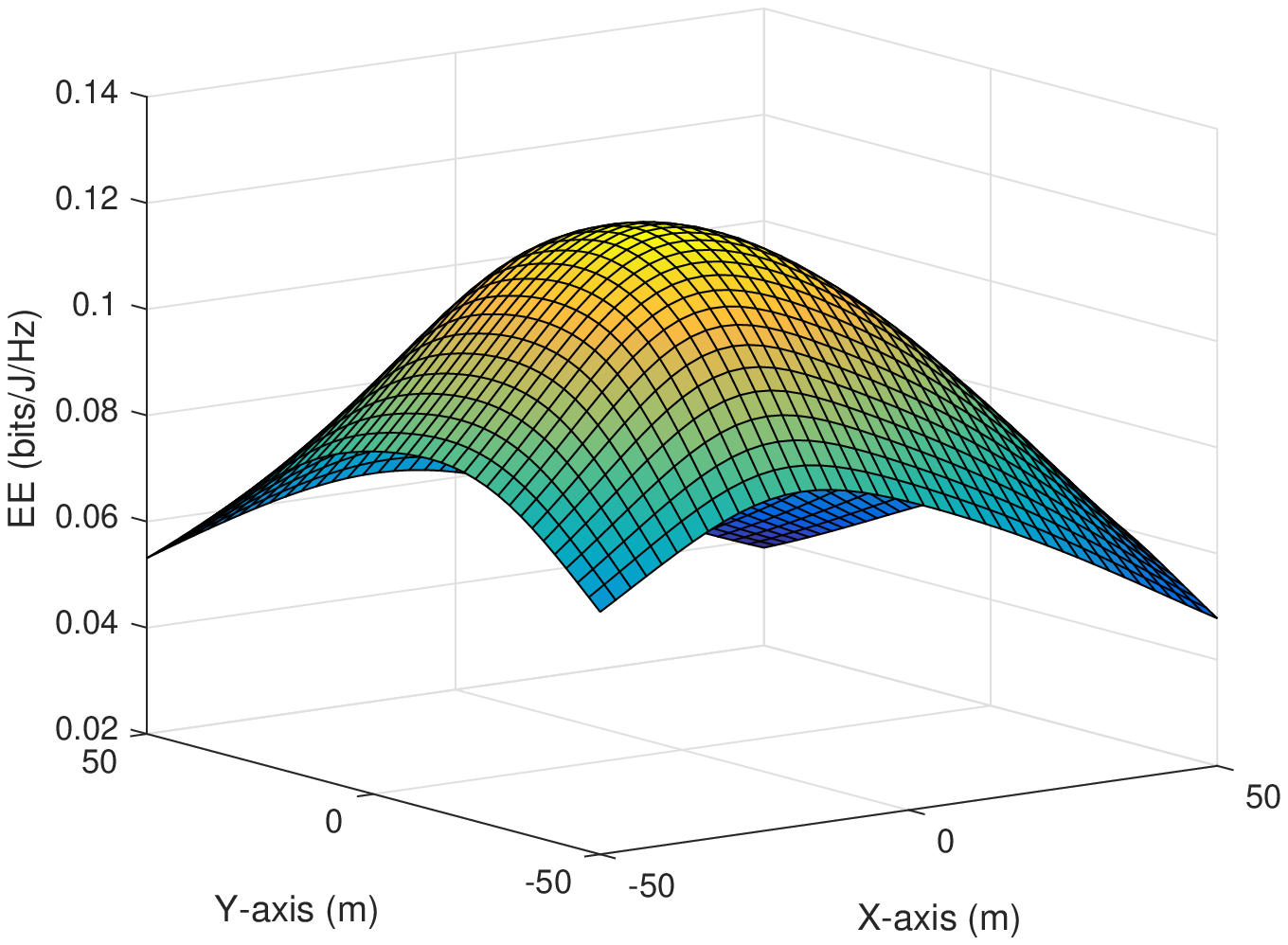}
\end{minipage}
}%
\subfigure[EE performance versus UAV location of 20 user systems.]{
\begin{minipage}[t]{0.33\linewidth}
\centering
\includegraphics[width=2in]{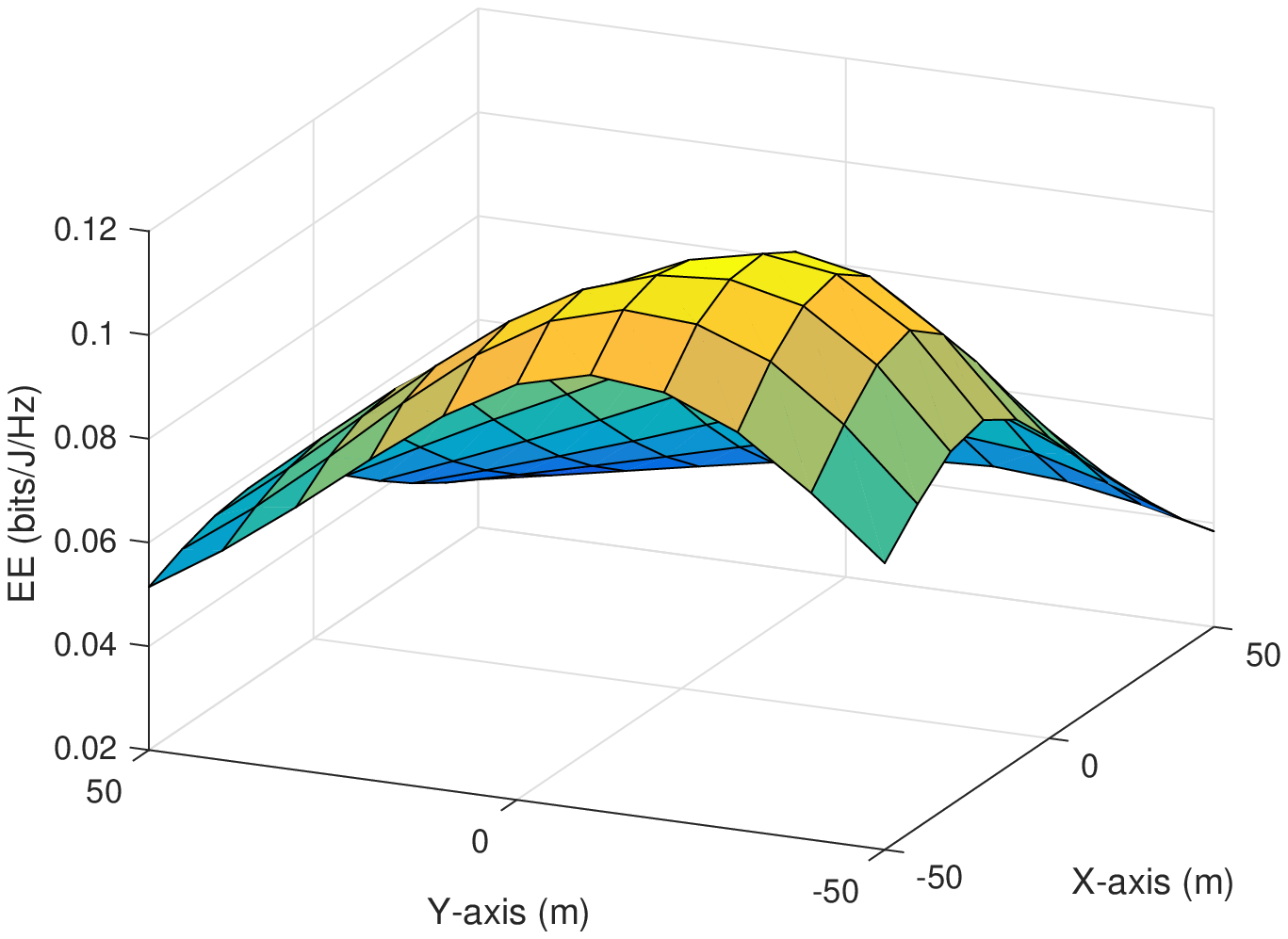}
\end{minipage}
}%
\subfigure[EE performance versus UAV location of 30 user systems.]{
\begin{minipage}[t]{0.33\linewidth}
\centering
\includegraphics[width=2in]{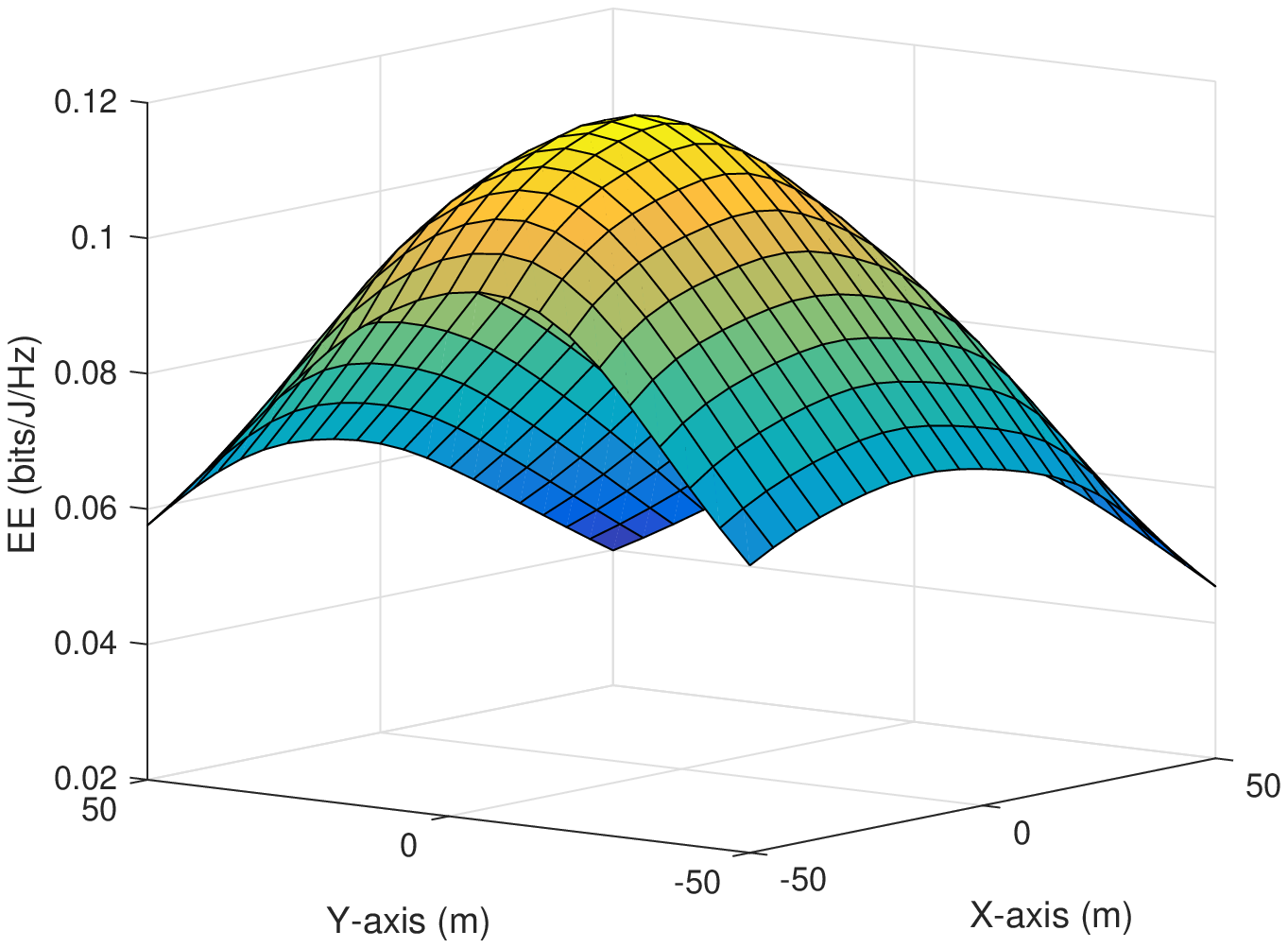}
\end{minipage}%
}%

\centering
\caption{ The EE of the proposed SMG-based RA and LD algorithm under different UAV location. }
\end{figure*}

Fig. 12 plots the EE of the proposed SMG-based RA and LD algorithm, the ETPA-based RA and LD algorithm, and the exhaustive search (ES)-based algorithm \cite{Nievergelt}. With the increasing transmit power of the power beacon, the EE first increases and then decreases. This is consistent with Figs. 3 and 7. We also observe that the gap between the ES-based algorithm and the proposed algorithm gradually decrease. Since the ES-based algorithm offers the maximum EE at the cost of a significant time complexity, the small gap between the ES-based algorithm and the proposed algorithm indicates the effectiveness of the proposed algorithm.

\begin{figure}
\center
\includegraphics[width=8cm]{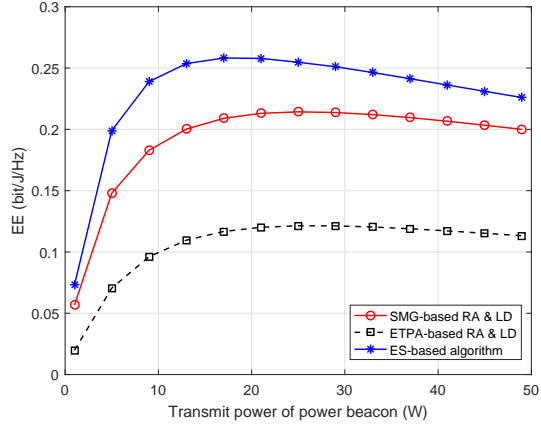}
\caption{The EE comparison of the proposed SMG-based RA and LD algorithm, the ETPA-based RA and LD algorithm, and the ES-based approach.}\vspace{-1.0em}
\end{figure}

\section{Conclusion}
In this paper, energy-efficient resource allocation was investigated in WPT-powered, UAV-assisted multiuser NOMA systems.
A joint optimization of the WPT time, UAV transmit power, NOMA power allocation coefficient, and UAV location was proposed to maximize the system EE.
To bypass the nonconvexity and intractability of this problem, we decomposed the problem into three subproblems solved iteratively following the BCD method.
We developed a supermodular game-based algorithm to reach a Nash equilibrium point of the intra-pair power allocation subproblem.
We jointly optimized the inter-pair power allocation and the WPT time allocation using the SCA method.
We utilized the Lagrange multiplier method to obtain the optimal UAV location.
Simulation showed the convergence and effectiveness of the proposed algorithm, and its superior EE performance, compared to the benchmark schemes.

\begin{appendices}
\section{Proof of Theorem 1}
We treat $f_{n,1}$ as the sum of $\frac{ \log _{2}\left(1+p_{n,1} |{h}_{n,1}|^2 \right) }{p_{n,1}+ P_{uav} + 2NP_{user} }$ and $-e^{\kappa p_{n,1}}$.
The concavity of $-e^{\kappa p_{n,1}}$ is obvious.
The second-derivatives of $\frac{ \log _{2}\left(1+p_{n,1} |{h}_{n,1}|^2 \right) }{p_{n,1}+ P_{uav} + 2NP_{user} }$ and $-e^{\kappa p_{n,1}}$ with respect to power $p_{n,1}>0$ are given by
\begin{small}
\begin{align}
&\frac{{{\partial }^{2}}\left( \frac{{{\log }_{2}} ( 1+{{p}_{n,1}}{{\left| {{h}_{n,1}} \right|}^{2}} )} {{ p_{n,1}+ P_{uav} + 2NP_{user} }}\; \right)}{\partial p_{n,1}^{2}}= \nonumber\\
&-\frac{{{\left| {{h}_{n,1}} \right|}^{4}}}{\log (2)({{P}_{uav}}+2N{{P}_{user}}+{{p}_{n,1}}){{({{\left| {{h}_{n,1}} \right|}^{2}}{{p}_{n,1}}+1)}^{2}}} \nonumber\\
 & -\frac{2{{\left| {{h}_{n,1}} \right|}^{2}}}{\log (2){{({{P}_{uav}}+2N{{P}_{user}}+{{p}_{n,1}})}^{2}}({{\left| {{h}_{n,1}} \right|}^{2}}{{p}_{n,1}}+1)} \nonumber \\
 &+\frac{2\log ({{\left| {{h}_{n,1}} \right|}^{2}}{{p}_{n,1}}+1)}{\log (2){{({{P}_{uav}}+2N{{P}_{user}}+{{p}_{n,1}})}^{3}}} <0 \\
& \frac{{{\partial }^{2}}\left( -{{e}^{\kappa {{p}_{n,1}}}} \right)}{\partial p_{n,1}^{2}}=-{{\kappa }^{2}}{{e}^{\kappa {{p}_{n,1}}}}<0. \ \ \ \ \ \ \ \ \ \ \ \ \ \
\end{align}
\end{small}
As a result, the concavity of $ f_{n,1}$ is proved.
If ${p}_{n,1}$ is fixed, we simplify $\frac{{{\left| {{h}_{n,2}} \right|}^{2}}}{ 1 + {{p}_{n,1}}{{\left| {{h}_{n,2}} \right|}^{2}}}$ to $|\hbar_{n,2}|^2$,
and then the utility function $f_{n,2}=\frac{ \log _{2}\left(1+p_{n,2} |\hbar_{n,2}|^2 \right) }{p_{n,2}+ P_{uav} + 2NP_{user} }-e^{\kappa p_{n,2}}$, which has the same form as $f_{n,1}$.

Therefore, the concavity of $f_{n,k}$ is proved.
The convexity proof of Theorem 1 is complete. \ \ \ \ \ \ \ \ \ \ \ \ \ \ \ \ \ \ \ \ \ \ \ \ \ \ \ \ \ \ \qedsymbol

\section{Proof of Theorem 3}
The strategy space is obviously nonempty. Moreover, the
strategy space is a subset of the $\mathbb{R}$ space, so it is a compact sublattice.
Therefore, the first condition is satisfied.
Additionally, the functions $f_{n,1}$ and $f_{n,2}$ are twice differentiable.
The partial derivative of $f_{n,1}$ with respect to $p_{n,1}$ and $p_{n,2}$ is written as
$ \frac{\partial^{2} f_{n,1}}{\partial p_{n,1} \partial p_{n,2}}=0$.
Therefore, it meets the second condition in Definition 3 regarding the supermodular game. The partial derivatives of $f_{n,2}$ with respect to $p_{n,2}$ and $p_{n,1}$ are given by
\begin{small}
\begin{equation}
\begin{aligned}
& \ \ \ \frac{\partial f_{n,2}}{\partial p_{n,2}}=\frac{\partial \left( {{{\log }_{2}}\left( \Upsilon  \right)}/{\left( {{p}_{n,2}}+P_s \right)}\; \right)}{\partial {{p}_{n,2}}} \\
&=\frac{\frac{{{\left| {{h}_{n,2}} \right|}^{2}}\left( {{p}_{n,2}}+P_s \right)}{{{\left| {{h}_{n,2}} \right|}^{2}}\left( {{p}_{n,1}}+{{p}_{n,2}} \right)+ 1 }-\ln \left( \Upsilon  \right)}{\ln 2{{\left( P_s +{{p}_{n,2}} \right)}^{2}}}-\kappa {{e}^{\kappa {{p}_{n,2}}}}
\end{aligned}
\end{equation}
\end{small}
\begin{small}
\begin{equation}
\begin{aligned}
  &\ \ \ \frac{\partial^{2} f_{n,2}}{\partial p_{n,2} \partial p_{n,1}} =  \frac{{{\partial }^{2}}\left( {{{\log }_{2}}\left( \Upsilon  \right)}/{\left( {{p}_{n,2}}+{{P}_{uav}}+2N{{P}_{user}} \right)}\; \right)}{\partial {{p}_{n,2}}\partial {{p}_{n,1}}} \\
 & =\frac{{{\left| {{h}_{n,2}} \right|}^{6}}{{p}_{n,2}}}{\ln 2({{P}_{uav}}+2N{{P}_{user}}+{{p}_{n,2}}){{\left( {{\left| {{h}_{n,2}} \right|}^{2}}{{p}_{n,1}}+ 1 \right)}^{3}}{{\Upsilon }^{2}}} \\
 & -\frac{{{\left| {{h}_{n,2}} \right|}^{4}}}{\ln 2({{P}_{uav}}+2N{{P}_{user}}+{{p}_{n,2}}){{\left( {{\left| {{h}_{n,2}} \right|}^{2}}{{p}_{n,1}}+ 1 \right)}^{2}}\Upsilon }  \\
  & +\frac{{{\left| {{h}_{n,2}} \right|}^{4}}{{p}_{n,2}}}{\ln 2{{({{P}_{uav}}+2N{{P}_{user}}+{{p}_{n,2}})}^{2}}{{\left( {{\left| {{h}_{n,2}} \right|}^{2}}{{p}_{n,1}}+ 1 \right)}^{2}}\Upsilon } \\
  &=\frac{{{\left| {{h}_{n,2}} \right|}^{4}}{{\left( -{{P}_{s}}{{\left| {{h}_{n,2}} \right|}^{2}}{{p}_{n,1}}-{{P}_{s}} +{{\left| {{h}_{n,2}} \right|}^{2}}p_{n,2}^{2} \right)}^{2}}}{\ln 2{{({{P}_{s}}+{{p}_{n,2}})}^{2}}\left( {{\left| {{h}_{n,2}} \right|}^{2}}{{p}_{n,1}}+ 1 \right) \Upsilon_1    },
\end{aligned}
\end{equation}
\end{small}
where $\Upsilon =\frac{{{\left| {{h}_{n,2}} \right|}^{2}}{{p}_{n,2}}}{{{\left| {{h}_{n,2}} \right|}^{2}}{{p}_{n,1}}+ 1 }+1$, $P_s = {{P}_{uav}}+2N{{P}_{user}}$ and $\Upsilon_1 = {{\left( {{\left| {{h}_{n,2}} \right|}^{2}}{{p}_{n,2}}+{{\left| {{h}_{n,2}} \right|}^{2}}{{p}_{n,1}}+ 1 \right)}^{2}}$.
Let $\frac{\partial^{2} f_{n,2}}{\partial p_{n,2} \partial p_{n,1}} \geq 0$, we can obtain
\begin{small}
\begin{equation}
 p_{n,2} \geq \sqrt{\frac{({P_{uav}+2NP_{user} })\left(p_{n,1} {{\left| {{h}_{n,2}} \right|}^{2}} + 1  \right)}{ {{\left| {{h}_{n,2}} \right|}^{2}}} }.
\end{equation}
\end{small}
Whenever the above inequality holds, the second condition is satisfied and
our proposed game is a supermodular game.
The proof of Theorem 3 is complete.
\ \ \ \ \ \ \ \ \ \ \ \ \ \ \ \ \ \ \ \ \ \ \ \ \ \ \ \ \ \ \ \ \ \ \qedsymbol

\section{Convergence Proof of Algorithm 4}

Define $\Xi_{EE,p_n,\tau}^{(r),lb} (\{\alpha_{n,k}\} ,\{{p_n},\tau\}, \{{x_0},{y_0}\}) = \Xi_{EE,p_n,\tau}^{(r)}$, where $\Xi_{EE,p_n,\tau}^{(r)}$ is the objective value of problem (\ref{equP3}) at the initialization of $\{\alpha_{n,k}\}$, $\{{p_n},\tau\}$, and $\{{x_0},{y_0}\}$.
In step 3 of Algorithm 4, since the optimal solution to problem (\ref{equP2}) is obtained based on given $\{{p_n}^{(r)},{\tau}^{(r)}\}$, and $\{{x_0}^{(r)},{y_0}^{(r)}\}$, we have
\begin{small}
\begin{equation}
\label{EQ:AP4-1}
\begin{aligned}
& \ \ \ \ \Xi_{EE} (\{\alpha_{n,k}^{(r)}\} ,\{{p_n}^{(r)},\tau^{(r)}\}, \{{x_0}^{(r)},{y_0}^{(r)}\}) \\
 &\leq \Xi_{EE} (\{\alpha_{n,k}^{(r+1)}\} ,\{{p_n}^{(r)},\tau^{(r)}\}, \{{x_0}^{(r)},{y_0}^{(r)}\}).
\end{aligned}
\end{equation}
\end{small}

For given $\{\alpha_{n,k}^{(r+1)}\}$, $\{{p_n}^{(r)}, \tau^{(r)}\}$, and $\{{x_0}^{(r)},{y_0}^{(r)}\}$ in step 4, it follows that
\begin{small}
\begin{equation}
\label{EQ:AP4-2}
\begin{aligned}
& \ \ \ \ \Xi_{EE} (\{\alpha_{n,k}^{(r)}\} ,\{{p_n}^{(r)},\tau^{(r)}\}, \{{x_0}^{(r)},{y_0}^{(r)}\}) \\
&\overset{(a)}{\mathop{\le }} \ \Xi_{EE,p_n,\tau}^{(r),lb} (\{\alpha_{n,k}^{(r+1)}\} ,\{{p_n}^{(r+1)},\tau^{(r+1)}\}, \{{x_0}^{(r)},{y_0}^{(r)}\}) \\
&  \overset{(b)}{\mathop{\le }} \ \Xi_{EE} (\{\alpha_{n,k}^{(r+1)}\} ,\{{p_n}^{(r+1)},\tau^{(r+1)}\}, \{{x_0}^{(r)},{y_0}^{(r)}\}),
\end{aligned}
\end{equation}
\end{small}

\noindent where (a) holds, as problem (\ref{equP3}) is solved optimally with the solution $\{{p_n}^{(r+1)},\tau^{(r+1)}\}$ in step 4 given $\{\alpha_{n,k}^{(r+1)}\}$ and $\{{x_0}^{(r)},{y_0}^{(r)}\}$, and (b) holds, as the objective value of problem (\ref{equPT4}) is the lower bound for the original problem (\ref{equP3}) at $\{{p_n}^{(r+1)},\tau^{(r+1)}\}$.

Finally, given $\{\alpha_{n,k}^{(r+1)}\}$, $\{{p_n}^{(r+1)},\tau^{(r+1)}\}$ and  $\{{x_0}^{(r)},{y_0}^{(r)}\}$, we have
\begin{small}
\begin{equation}
\label{EQ:AP4-3}
\begin{aligned}
 & \ \ \ \ \Xi_{EE} (\{\alpha_{n,k}^{(r+1)}\} ,\{{p_n}^{(r+1)},\tau^{(r+1)}\}, \{{x_0}^{(r)},{y_0}^{(r)}\}) \\
 &\leq \Xi_{EE} (\{\alpha_{n,k}^{(r+1)}\} ,\{{p_n}^{(r+1)},\tau^{(r+1)}\}, \{{x_0}^{(r+1)},{y_0}^{(r+1)}\}).
\end{aligned}
\end{equation}
\end{small}
Based on (\ref{EQ:AP4-1})-(\ref{EQ:AP4-3}), we can obtain
\begin{small}
\begin{equation}
\label{EQ:AP4-final}
\begin{aligned}
& \ \ \ \ \Xi_{EE} (\{\alpha_{n,k}^{(r)}\} ,\{{p_n}^{(r)},\tau^{(r)}\}, \{{x_0}^{(r)},{y_0}^{(r)}\})\\
& \leq \Xi_{EE} (\{\alpha_{n,k}^{(r+1)}\} ,\{{p_n}^{(r+1)},\tau^{(r+1)}\}, \{{x_0}^{(r+1)},{y_0}^{(r+1)}\}),
\end{aligned}
\end{equation}
\end{small}
which shows that the objective value of (\ref{equP1}) is nondecreasing after each iteration of Algorithm 4. The convergence of Algorithm 4 is proved.
\ \ \ \ \ \ \ \ \ \ \ \ \ \ \ \ \ \ \ \ \ \ \ \ \ \ \ \ \ \ \ \ \ \ \ \ \ \ \ \ \ \ \ \ \ \ \ \ \ \ \ \ \ \ \ \ \ \ \ \ \qedsymbol
\end{appendices}

\vspace{-1em}

\bibliographystyle{IEEEtran}
\bibliography{ciations}

\begin{IEEEbiography}[{\includegraphics[width=1in,height =1.25in,clip,keepaspectratio]{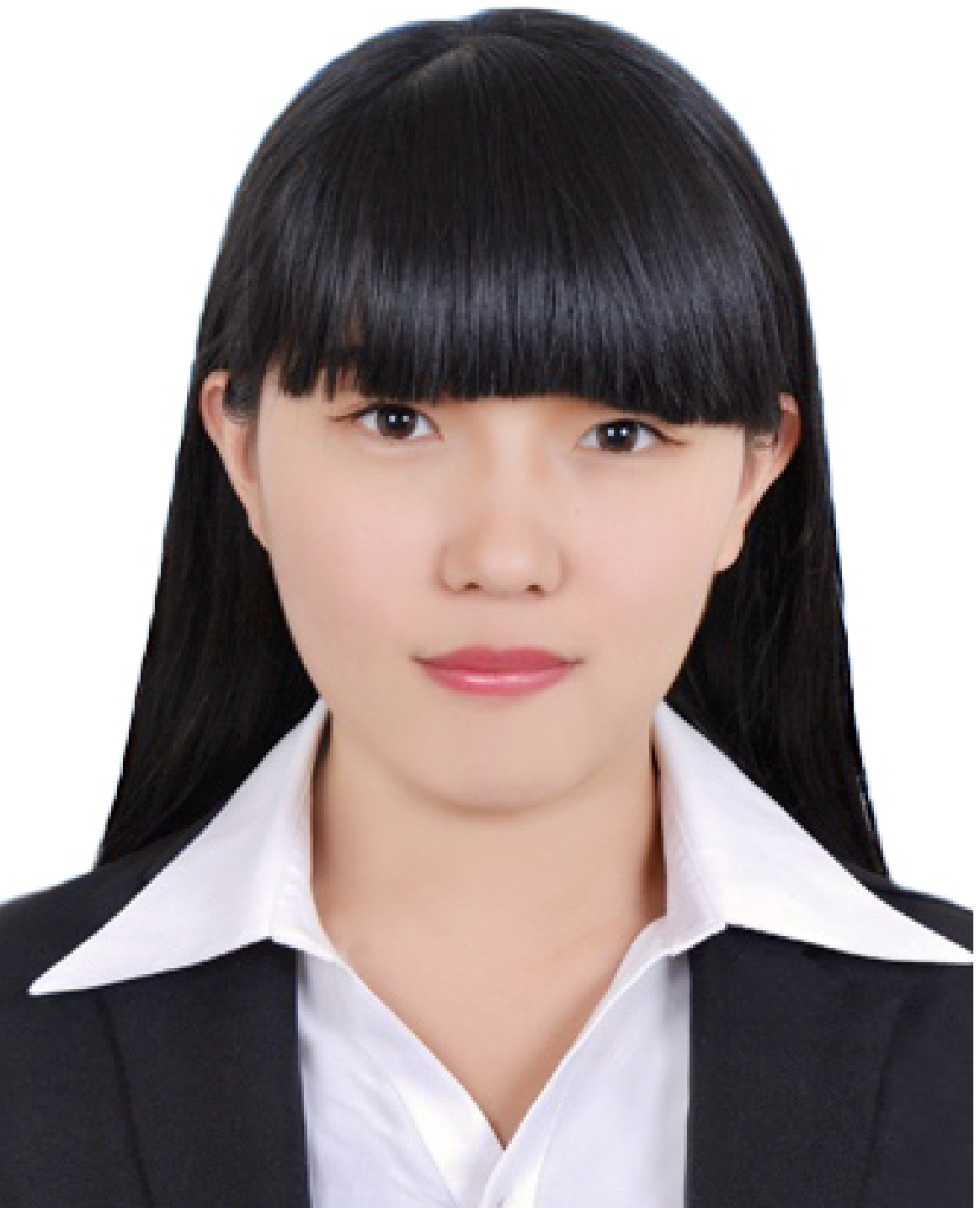}}]{Zhongyu Wang} (S'19) received the M.S. degree in computer science and technology from School of Information Science and Engineering, Yanshan University, Qinhuangdao, China. She is currently working toward the Ph.D. degree in information and communication engineering with the School of Information and Communication Engineering, Beijing University of Posts and Telecommunications, Beijing, China. Her current research interests include ultra-reliable low latency communication, non-orthogonal multiple access, massive MIMO, energy harvesting technology, and UAV communication.
\end{IEEEbiography}

\begin{IEEEbiography}[{\includegraphics[width=1in,height=1.25in,clip,keepaspectratio]{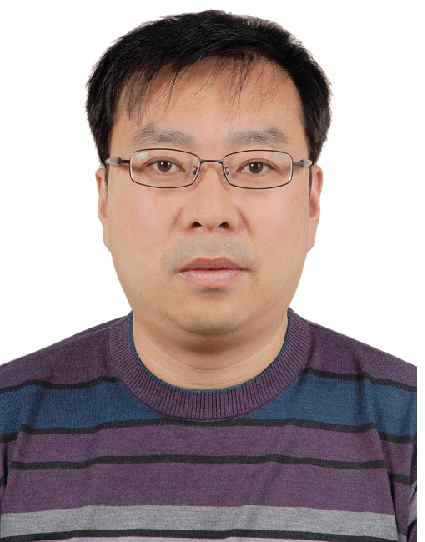}}]{Tiejun Lv}
(M'08-SM'12) received the M.S. and Ph.D. degrees in electronic engineering from the University of Electronic Science and Technology of China (UESTC), Chengdu, China, in 1997 and 2000, respectively. From January 2001 to January 2003, he was a Postdoctoral Fellow with Tsinghua University, Beijing, China. In 2005, he was promoted to a Full Professor with the School of Information and Communication Engineering, Beijing University of Posts and Telecommunications (BUPT). From September 2008 to March 2009, he was a Visiting Professor with the Department of Electrical Engineering, Stanford University, Stanford, CA, USA. He is the author of three books, more than 100 published IEEE journal papers and 200 conference papers on the physical layer of wireless mobile communications. His current research interests include signal processing, communications theory and networking. He was the recipient of the Program for New Century Excellent Talents in University Award from the Ministry of Education, China, in 2006. He received the Nature Science Award in the Ministry of Education of China for the hierarchical cooperative communication theory and technologies in 2015.
\end{IEEEbiography}

\begin{IEEEbiography}[{\includegraphics[width=1in,height=1.25in,clip,keepaspectratio]{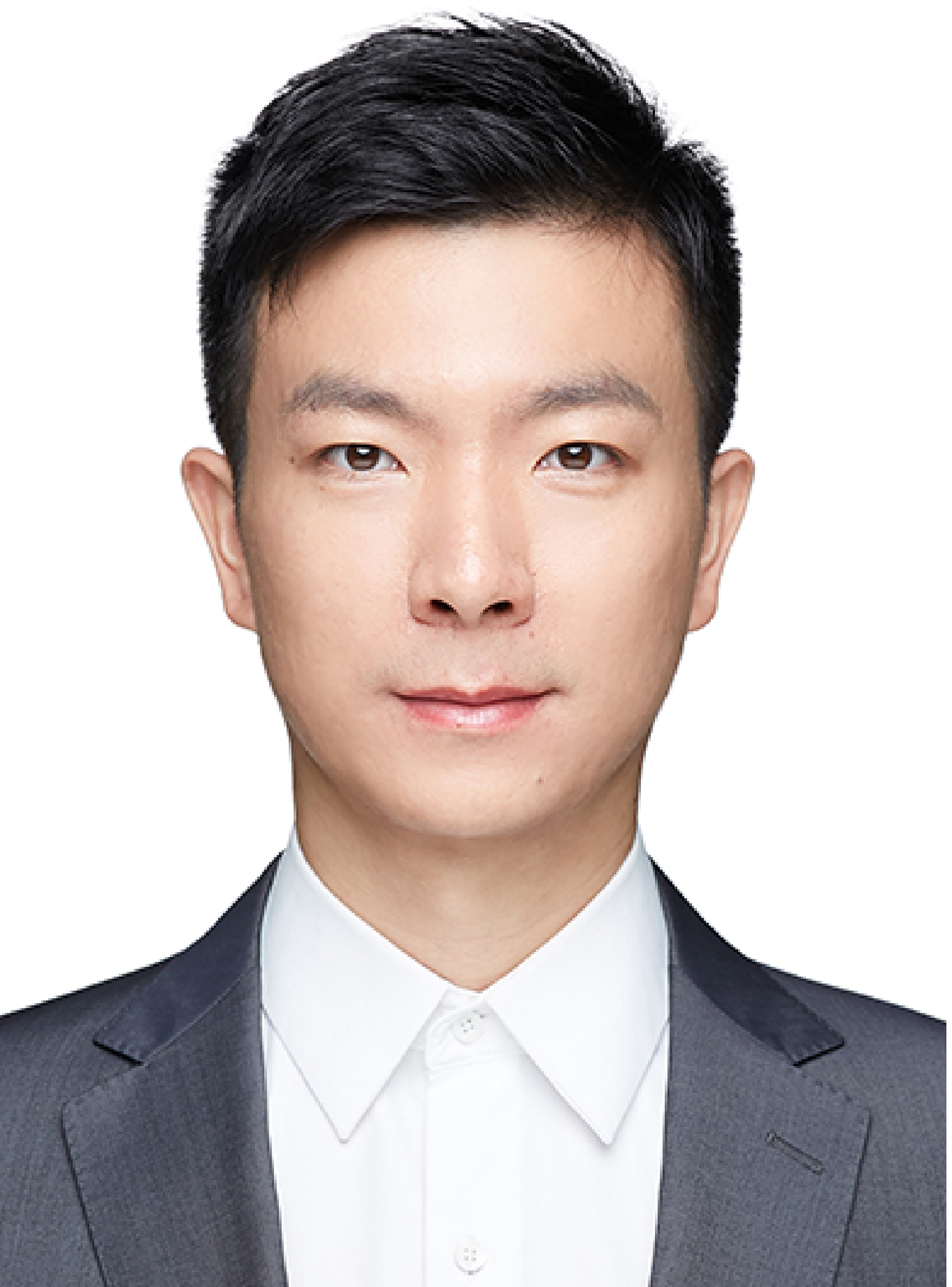}}]{Jie Zeng} (M'09-SM'16) received the B.S. and M.S. degrees from Tsinghua University in 2006 and 2009, respectively, and received two Ph.D. degrees from Beijing University of Posts and Telecommunications in 2019 and the University of Technology Sydney in 2021, respectively.

From July 2009 to May 2020, he was with the Research Institute of Information Technology, Tsinghua University. From May 2020 to April 2022, he was a postdoctoral researcher with the Department of Electronic Engineering, Tsinghua University. Since May 2022, he has been an associate professor with the School of Cyberspace Science and Technology, Beijing Institute of Technology.

His research interests include 5G/6G, URLLC, satellite Internet, and novel network architecture. He has published over 100 journal and conference papers, and holds more than 40 Chinese and international patents. He participated in drafting one national standard and one communication industry standard in China.

He received Beijing's science and technology award of in 2015, the best cooperation award of Samsung Electronics in 2016, and Dolby Australia's best scientific paper award in 2020.
\end{IEEEbiography}

\begin{IEEEbiography}[{\includegraphics[width=1in,height=1.25in,clip,keepaspectratio]{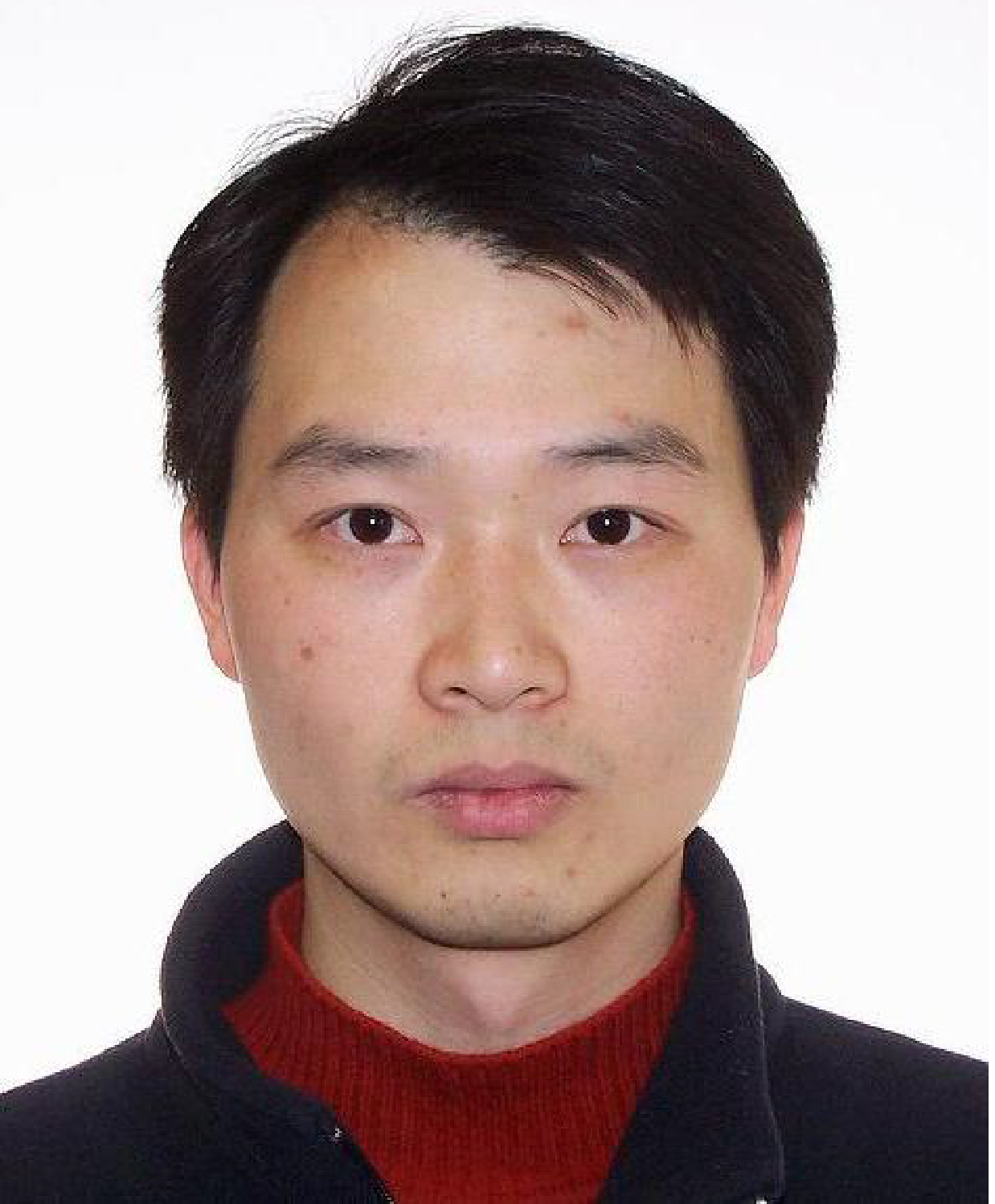}}]{Wei Ni} (M'09-SM'15) received the B.E. and Ph.D. degrees in Electronic Engineering from Fudan University, Shanghai, China, in 2000 and 2005, respectively. Currently, he is a Principal Research Scientist at CSIRO, Sydney, Australia, and an Adjunct Professor at the University of Technology Sydney and Honorary Professor at Macquarie University. He was a Postdoctoral Research Fellow at Shanghai Jiaotong University from 2005 -- 2008; Deputy Project Manager at the Bell Labs, Alcatel/Alcatel-Lucent from 2005 to 2008; and Senior Researcher at Devices R\&D, Nokia from 2008 to 2009. He has authored five book chapters, more than 200 journal papers, more than 80 conference papers, 25 patents, and ten standard proposals accepted by IEEE. His research interests include online learning, stochastic optimization, as well as their applications to system efficiency, integrity and privacy.

Dr Ni is the Chair of IEEE Vehicular Technology Society (VTS) New South Wales (NSW) Chapter since 2020, Editor of IEEE Transactions on Wireless Communications since 2018, and Editor of IEEE Transactions on Vehicular Technology since 2022. He served first the Secretary and then the Vice-Chair of IEEE NSW VTS Chapter from 2015 to 2019, Track Chair for VTC-Spring 2017, Track Co-chair for IEEE VTC-Spring 2016, Publication Chair for BodyNet 2015, and Student Travel Grant Chair for WPMC 2014.
\end{IEEEbiography}

\end{document}